\documentclass[12pt, letterpaper]{article}
\date{}
\usepackage{cite}
\usepackage{array}
\usepackage{graphicx}
\usepackage{epstopdf}
\usepackage{amsmath}
\usepackage{amssymb}
\usepackage{amsthm}
\def\mathbi#1{\textbf{\em #1}}
\usepackage{caption}
\usepackage{subcaption}
\usepackage[title]{appendix}
\usepackage[top=1in, bottom=1in, left=1in, right=1in]{geometry}
\usepackage{setspace}

\newtheorem{lemma}{Lemma}
\newtheorem{theorem}{Theorem}

\begin{document}
\title{Handover Count Based Velocity Estimation and Mobility State Detection in Dense HetNets}

\author{Arvind~Merwaday$^1$, and~\.{I}smail~G\"uven\c{c}$^2$\\
Email: {\tt $^1$amerw001@fiu.edu, $^2$iguvenc@fiu.edu}}
\maketitle

\begin{abstract}
In wireless cellular networks with densely deployed base stations, knowing the velocities of mobile devices is a key to avoid call drops and improve the quality of service to the user equipments (UEs). A simple and efficient way to estimate a UE's velocity is by counting the number of handovers made by the UE during a predefined time window. Indeed, handover-count based mobility state detection has been standardized since Long Term Evolution (LTE) Release-8 specifications. The increasing density of small cells in wireless networks can help in accurate estimation of velocity and mobility state of a UE. In this paper, we model densely deployed small cells using stochastic geometry, and then analyze the statistics of the number of handovers as a function of UE velocity, small-cell density, and handover count measurement time window. Using these statistics, we derive approximations to the Cramer-Rao lower bound (CRLB) for the velocity estimate of a UE. Also, we determine a minimum variance unbiased (MVU) velocity estimator whose variance tightly matches with the CRLB. Using this velocity estimator, we formulate the problem of detecting the mobility state of a UE as low, medium, or high-mobility, as in LTE specifications. Subsequently, we derive the probability of correctly detecting the mobility state of a UE. Finally, we evaluate the accuracy of the velocity estimator under more realistic scenarios such as clustered deployment of small cells, random way point (RWP) mobility model for UEs, and variable UE velocity. Our analysis shows that the accuracy of velocity estimation and mobility state detection increases with increasing small cell density and with increasing handover count measurement time window.
\end{abstract}

\noindent{\bf keywords: Cramer-Rao lower bound (CRLB), heterogeneous networks (HetNets), long term evolution (LTE), mobility state estimation, mobile velocity estimation, phantom cell, small cells.}

\doublespacing
\newpage
\section{Introduction}
Since the introduction of advanced mobile devices with data-intensive applications, cellular networks are witnessing rapidly increasing data traffic demands from mobile users. To keep up with the increasing traffic demands, cellular networks are being transformed into heterogeneous networks (HetNets) by the deployment of small cells (picocells, femtocells, etc) over the existing macrocells. Cisco has recently predicted an 11-fold increase in the global mobile data traffic between 2013 and 2018 \cite{2014-cisco-white-paper}, while Qualcomm has forecasted an astounding 1000x increase in mobile data traffic in the near future \cite{Qualcomm_1000xDataChallenge}. Addressing this challenge will lead to extreme densification of the small cells which will give rise to hyper-dense HetNets (HDHNs).

Mobility management in cellular networks is an important task that is critical to provide good quality of service to the mobile users by minimizing the handover failures. Velocity of a user equipment (UE) plays critical role on the handover performance of the UE particularly when the cell density is high~\cite{David_CM_2012,6215543,3gpp_HOF}, where knowing the UE's velocity becomes necessary for effective mobility management. In homogeneous networks that only have macrocells, handovers are typically finalized at the cell edge due to large cell sizes. With the deployment of small cell base stations (SBSs)~\cite{6831742,6952371,merwaday2014capacity,6171996}, due to smaller cell sizes, it becomes difficult to finalize the handover process at the cell edge for mobile devices~\cite{3GPP_Samsung_2010,David_CM_2012,6215543}. In particular, high-mobility devices may run deep inside the coverage areas of small cells before finalizing a handover, thus incurring handover failure due to degraded signal to interference plus noise ratio (SINR). These challenges motivate the need for UE-specific and cell-specific handover parameter optimization, which typically require estimation of the UE's velocity for effective configuration of handover parameters~\cite{Puttonen_ITNG_2009}. A UE's velocity estimate may also be used for scheduling~\cite{6661322,6706232}, mobility load  balancing~\cite{1001843}, channel quality indicator (CQI) feedback enhancements~\cite{SpeedAdaptiveCQI}, and energy efficiency enhancements~\cite{6692547}.

In this paper, we introduce a novel and efficient {\em handover-count}\footnote{Subsequently, \emph{``handover-count''} is used in a broad sense to refer to the number of cells traversed by a UE. E.g., in phantom cells~\cite{Ishii_GC_2012}, a UE is always connected to the MBSs, and can get additional throughput from SBSs, minimizing handover failures.} based UE velocity estimation technique using the tools from stochastic geometry, and characterize its accuracy through Cramer-Rao Lower Bounds (CRLBs), when the density of SBSs is known. Since the service provider has the information of the number of SBSs in a particular geographic area, the SBS density in that area can be calculated and broadcasted as part of system information in next generation networks. The SBS density may also be signaled in a user-specific manner to the next generation UEs which are capable of velocity estimation. Our contributions in this paper are as follows: 1) for a given small-cell density, two approximations to the probability mass function (PMF) of handover-count of a UE are derived using a heuristic approach; 2) using the PMF approximations, expressions for the CRLB of velocity estimation are derived; 3) a minimum variance unbiased (MVU) velocity estimator is derived whose variance tightly matches with the CRLB, and accuracy of the estimator is investigated for various UE speeds, SBS densities, and handover-count measurement times; 4) the estimated velocity is used to detect the mobility state (low/medium/high) of a UE, and the expressions for the probability of detection and probability of false alarm are derived; 5) accuracy of the velocity estimator is analyzed for different realistic scenarios such as: RWP mobility model for the UEs, clustered deployment of SBSs, and variable UE velocity.

This paper is organized as follows. We briefly review the related works and the state of the art in Section~\ref{sec:PriorWork}. In Section~\ref{sec:SystemModel}, we describe our system model of small-cells using stochastic geometry. Subsequently, we derive the approximations to the PMF of handover-count of a UE in Section~\ref{sec:ApproxPMF}. In Section~\ref{sec:CRLB}, we find CRLB for the velocity estimation of UE and also derive a MVU velocity estimator. In Section~\ref{sec:MobStateProb}, we provide expressions for the mobility state (low, medium, high) probabilities, the probability of detection, and the probability of false alarm. We show numerical results on our findings in Section~\ref{sec:NumResults} and conclude the paper in Section~\ref{sec:Conclusion}.

\section{Review of Prior Work}
\label{sec:PriorWork}
Existing long term evolution (LTE) and LTE-Advanced technologies are capable of estimating the mobility state of a UE into three broad classes: low, medium, and high-mobility~\cite{Speed_Differentiated_HO_2012,David_CM_2012,6215543,Puttonen_ITNG_2009}. This is achieved at a UE by counting the number of handovers within a given time window, and comparing it with a threshold which can be implemented during the connected mode~\cite{3GPPa} or the idle mode~\cite{3GPPTS36.3042009} of the UE. It can also be implemented at the network side, by tracking the prior history of handovers for a particular UE. The accuracy of a UE's mobility state estimate will benefit significantly from the densification of SBS deployment. While more accurate UE-side speed estimation techniques based on Doppler estimation have been discussed in~\cite{Iwamura_Patent_2010,510946,DopplarFreqEstimation,DopplarFreqEstimation_RecursiveMLE}, due to their complexity and standardization challenges, they have not been adopted in existing cellular network standards. While global positioning system (GPS) can be used for accurate estimation of a UE's speed, it may not be a practical solution for mobility management as 1)~the GPS receiver at a UE consumes significant amount of power, 2)~GPS coverage may not be available in environments such as urban canyons and underground subways, and 3) not all the UEs are equipped with GPS receivers.  A popular usecase example that uses such real-time mobility state information at a UE is that of UE-specific cell selection: high speed UEs can be \emph{biased}  to stay connected to  macrocells even when small cell link quality is better~\cite{pang2014method}.

There are various studies available in the literature on handover-count based mobility state detection (MSD). A simulation based mobility analysis is performed in \cite{Speed_Differentiated_HO_2012} where it is proposed that the number of handovers made by a UE are weighted differently for macro to macro, macro to pico, pico to macro, and pico to pico handovers (1, 0.45, 0.25, and 0.1, respectively), to produce a good estimate of UE mobility. In \cite{MobilityJeff}, considering a random way point (RWP) model and using stochastic geometry, the expected number of handovers during the movement period of a RWP model is derived. The probability density function (PDF) of the sojourn time is also presented for homogeneous networks. In \cite{WeiBaoHandoffRateAnalysis}, a theoretical framework using stochastic geometry is developed to study the UE mobility in HetNets, in which the expressions for vertical and horizontal handoff rates are derived. In \cite{7248855}, closed form expressions of cross-tier handover rate and sojourn time in a small cell are provided using stochastic geometry for a two-tier network. In another study~\cite{7006787}, stochastic geometry is used to derive the handover rate, which in turn is used to derive coverage probability of a UE in HetNet by considering that the UE is mobile and a fraction of the handovers result into failure.

A set of related works in the literature have studied the path prediction of UEs and the handoff time estimation along the predicted path which can help in improving the quality of service to the UEs. A destination and mobility path prediction model called DAMP is proposed in \cite{6872550}, while a handoff time window estimation method and mobility-prediction-aware bandwidth (MPBR) allocation scheme is presented in \cite{6870469}. In \cite{EfficientLocationPrediction}, different methods for predicting the future location of a UE based on prior knowledge of the UE's mobility are studied. A method for tracing UE's location using semi-supervised graph Laplacian approach is proposed in \cite{5989824}.

Despite the earlier work in~\cite{Speed_Differentiated_HO_2012, 3GPPa, David_CM_2012, 6872550, 6870469, EfficientLocationPrediction, 5989824, 6215543, Puttonen_ITNG_2009, 3GPPTS36.3042009, Iwamura_Patent_2010, 510946, DopplarFreqEstimation, DopplarFreqEstimation_RecursiveMLE, MobilityJeff, WeiBaoHandoffRateAnalysis, 7006787}, fundamental performance bounds and optimum algorithms for handover-count based velocity estimation have not been studied in the literature. The focus of this paper is on velocity estimation and MSD based on the handover counts of a UE; path prediction and handoff time estimation techniques are not considered. Mobility state estimation itself is an important topic in 3GPP Release-8 specifications, and being researched currently. Additionally, there are applications which require velocity estimation at the UE side, in which case, implementing the path prediction algorithms at a UE would be difficult. To our best knowledge, velocity estimation based on handover counts at a UE has not been studied analytically in the literature, which is one of the main contribution of this paper.

\section{System Model}
\label{sec:SystemModel}
Consider a HetNet scenario where the macro base stations (MBSs) tier and the SBSs tier use different frequency bands, as in phantom cell architecture~\cite{Ishii_GC_2012}. We assume that the SBSs are randomly distributed according to a homogeneous Poisson point process (PPP) with intensity $\lambda$. As the SBSs use a dedicated frequency band which is different from the frequency band used by the MBSs, the coverage of a small cell would depend on the neighboring SBSs, but not on MBSs. Hence, the small cells of the network can be modeled using Poisson-Voronoi tessellation as illustrated in Fig.~\ref{fig:Poisson_Line_Intersect}(b). In Fig.~\ref{fig:Poisson_Line_Intersect}(a), co-channel deployment of small cells is illustrated in which the coverage of small cells is determined by the interference from neighboring SBSs and MBSs. In this paper, we consider only the dedicated channel deployment of small cells as illustrated in Fig.~\ref{fig:Poisson_Line_Intersect}(b).
\begin{figure}[htp]
\vspace{-3mm}
\centering
\begin{subfigure}[b]{0.495\textwidth}
\includegraphics[width=\textwidth]{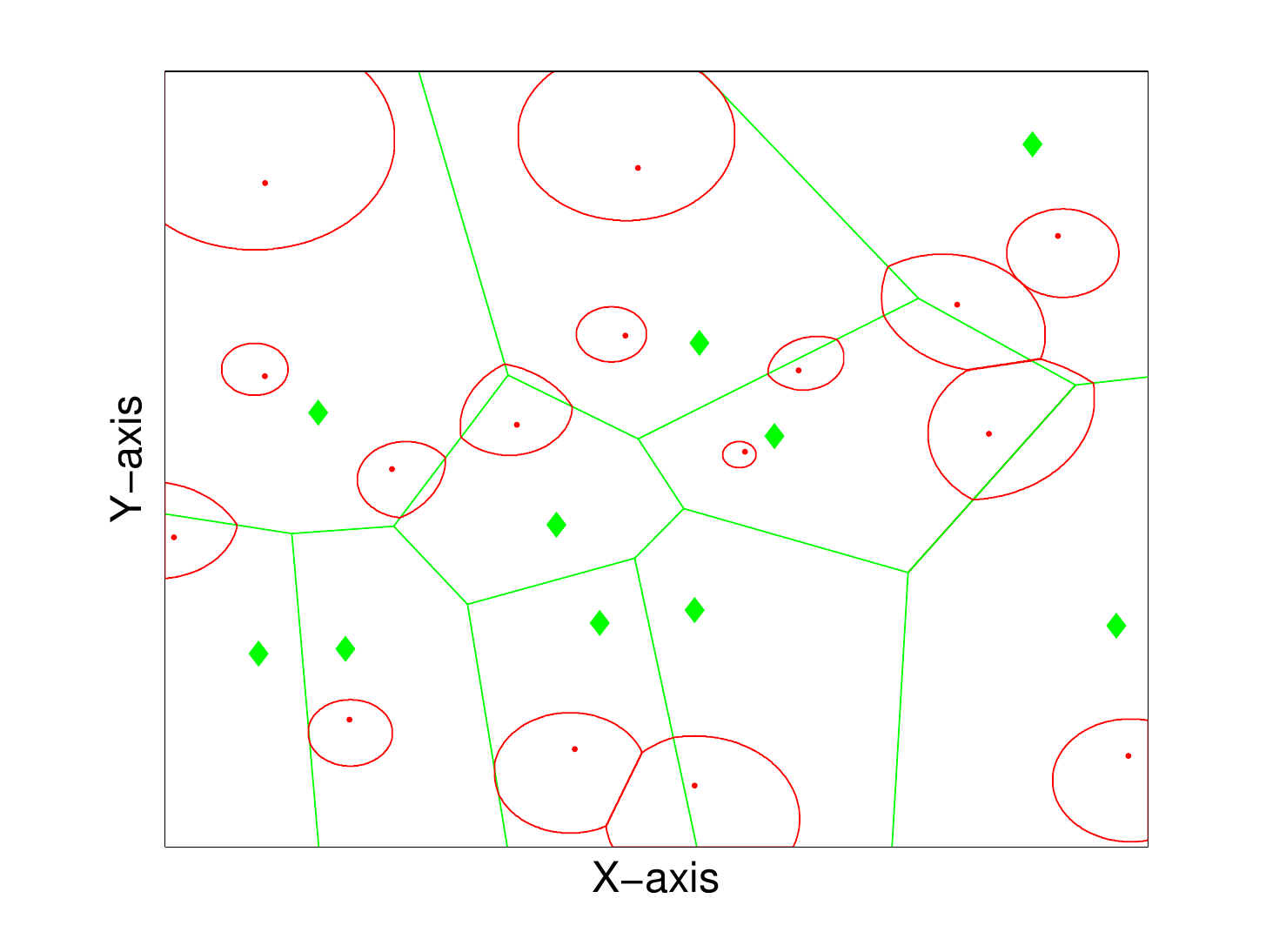}
\vspace{-8mm}\caption{}
\end{subfigure}
\begin{subfigure}[b]{0.495\textwidth}
\includegraphics[width=\textwidth]{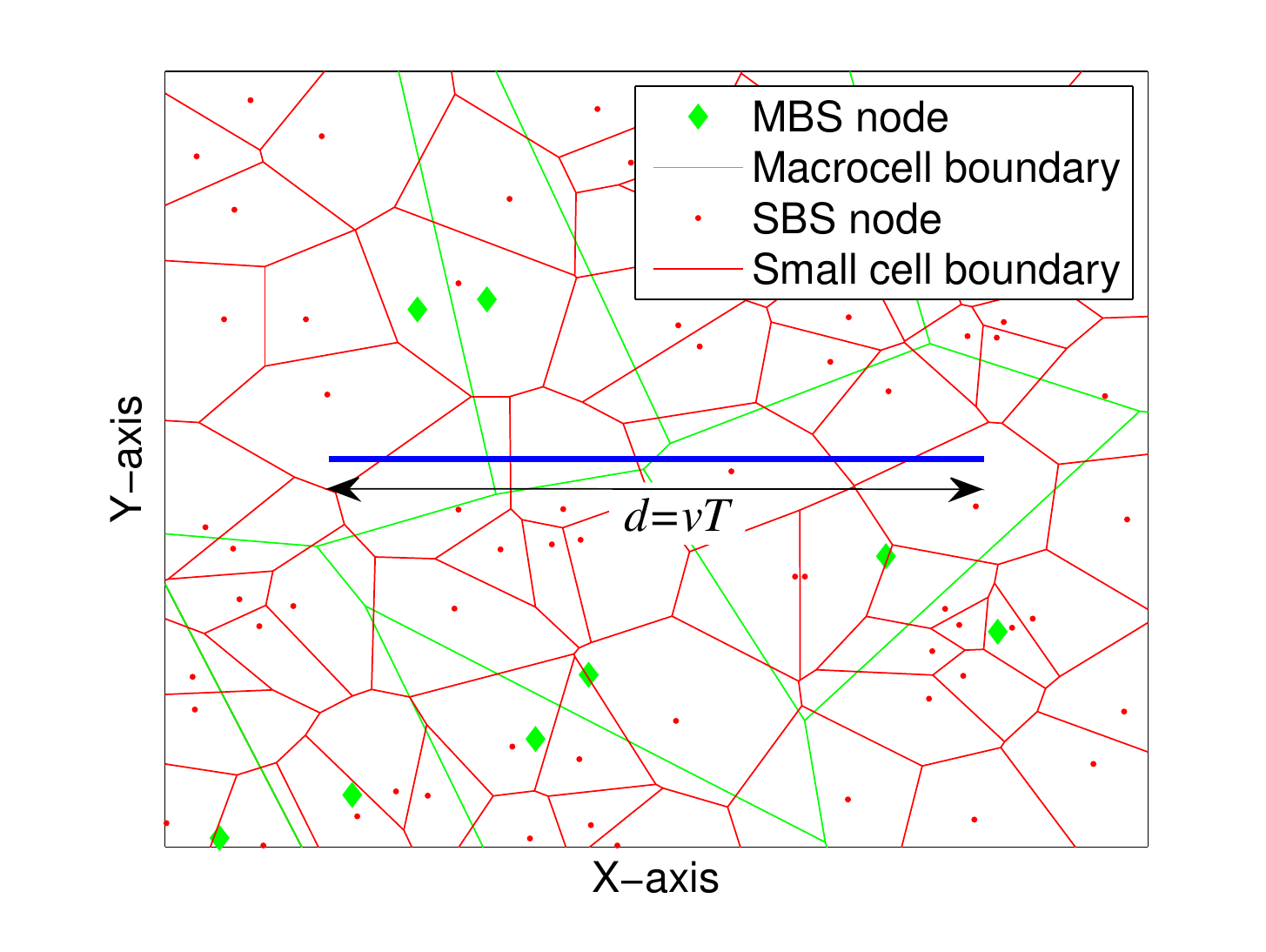}
\vspace{-8mm}\caption{}
\end{subfigure}
\caption{Illustration of cellular network layout; (a)~coverage of macrocells and small cells in co-channel deployment, as in a typical HetNet; (b)~coverage of small-cells with dedicated channel deployment as in a phantom cell network.}
\label{fig:Poisson_Line_Intersect}
\vspace{-5mm}
\end{figure}

We consider a simple UE mobility scenario as also shown in Fig.~\ref{fig:Poisson_Line_Intersect}(b), in which the UE travels along a linear trajectory (for example, through X-axis) with constant velocity $v$. During the travel, we assume that the UE can determine whenever it crosses the boundary of a small cell. In a broad sense, we call the boundary crossings made by the UE as handovers. Therefore, the number of handovers $H$ made by the UE during a measurement time window $T$ is equal to the number of intersections between the UE travel trajectory (of length $d = vT$) and the small-cell boundaries. We use linear mobility model for its simplicity in theoretical analysis, and this model is suitable for scenarios such as medium/high speed cars and trains that may travel through downtown areas. There may be many small cells deployed in such urban areas in the future which may be referred to as ultra/hyper-dense networks~\cite{NokiaUltraDenseNetworks}. Linear mobility may not be accurate for some other scenarios, therefore we have also considered RWP mobility model which is more general and includes linear mobility as a special case.

Note that the approach in Fig.~\ref{fig:Poisson_Line_Intersect}(b) for handover-count based MSD of a UE, into low/medium/high mobility states, has already been specified in the LTE Release-8 standard. Remarkably, no studies exist in the literature that investigate fundamental bounds and effective estimators for UE velocity. At high UE speeds and high SBS densities, handover failure of a UE becomes more likely \cite{3gpp_HOF,6675473,6364722}. Hence, the velocity estimation of UE based on handover counts may not work effectively  with conventional LTE mobility management. On the other hand, emerging small-cell architectures such as phantom cells \cite{Ishii_GC_2012,DocomoPhantomCell} decouple the control and user planes, and allow the UE to be connected to macrocell all the time. Small-cells are discovered through special discovery signals \cite{3gppSCdiscovery}, and a UE can connect them (reminiscent to a handover) to have higher throughput. Hence, handover-count based velocity estimation is still applicable in such scenarios.

\begin{figure}[t]
\centering
\begin{subfigure}[b]{0.35\textwidth}
\includegraphics[width=\textwidth]{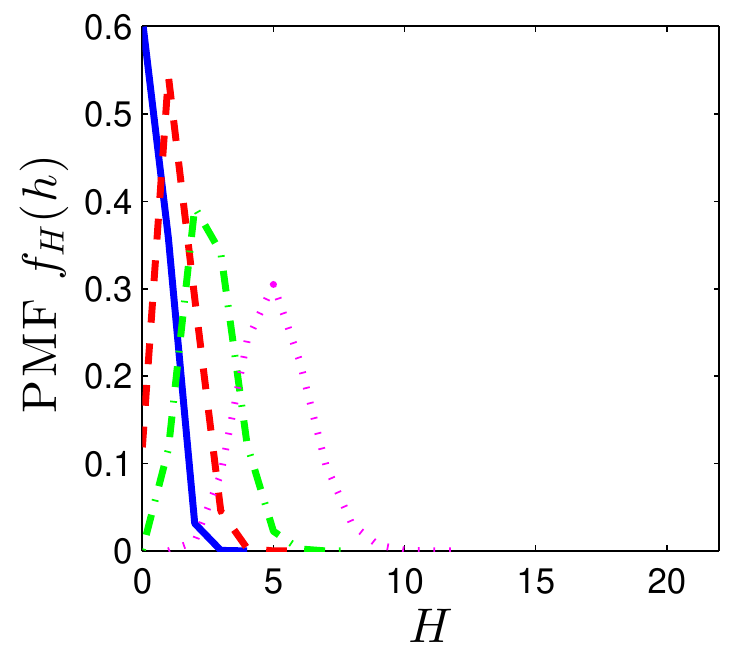}
\vspace{-7mm}\caption{}
\end{subfigure}
\begin{subfigure}[b]{0.35\textwidth}
\includegraphics[width=\textwidth]{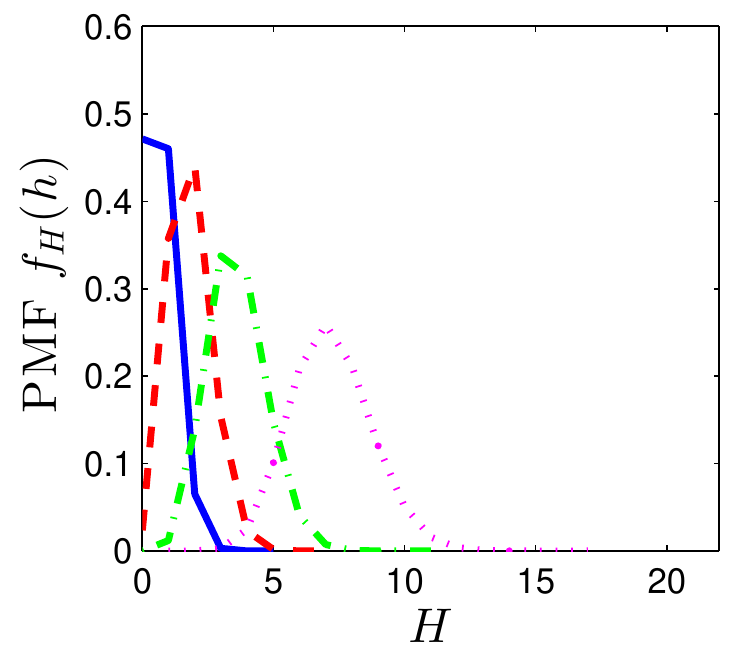}
\vspace{-7mm}\caption{}
\end{subfigure}
\begin{subfigure}[b]{0.35\textwidth}
\includegraphics[width=\textwidth]{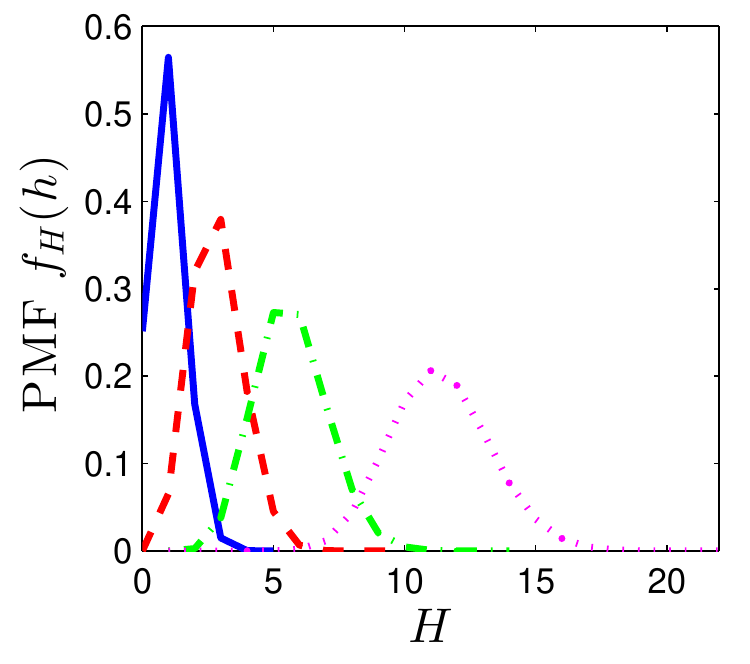}
\vspace{-7mm}\caption{}
\end{subfigure}
\begin{subfigure}[b]{0.35\textwidth}
\includegraphics[width=\textwidth]{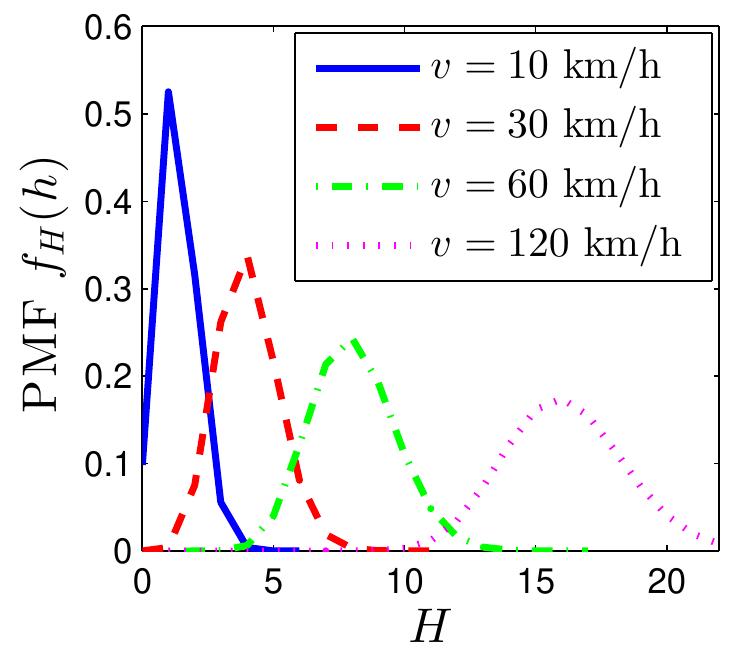}
\vspace{-7mm}\caption{}
\end{subfigure}
\caption{PMF of the handover count for different $\lambda$ and $v$ values; (a)~$\lambda=100$~SBSs/km$^2$; (b)~$\lambda=200$~SBSs/km$^2$; (c)~$\lambda=500$~SBSs/km$^2$; (d)~$\lambda=1000$~SBSs/km$^2$.}
\label{fig:NoHo_PMF}
\vspace{-3mm}
\end{figure}

The handover count $H$ for a scenario as in Fig.~\ref{fig:Poisson_Line_Intersect}(b) is a discrete random variable and its statistics do not change with the direction of the linear trajectory because the SBS locations are modeled using a homogeneous and stationary PPP. The probability mass function (PMF) $f_H(h)$ of handover count obtained through the simulations is shown in Fig.~\ref{fig:NoHo_PMF}(a)-\ref{fig:NoHo_PMF}(d) for different SBS density $\lambda$ and UE velocity $v$ settings, with the handover-count measurement time fixed to $T=12$~s. For low SBS densities, as in Fig.~\ref{fig:NoHo_PMF}(a), the PMFs for different UE velocities are overlapping significantly, leading to low velocity estimation accuracies. For higher SBS densities, as in Fig.~\ref{fig:NoHo_PMF}(d), the PMFs for different velocities are separated, leading to better estimation accuracies. It should also be noticed that for higher UE velocities, standard deviation of the handover count also increases, implying lower estimation accuracies. Another characteristic of the PMF is that for large $\lambda$ and $v$ values, the shape of PMF resembles Gaussian distribution.

\subsection{Modeling Handover-Count Statistics}
In order to obtain the velocity estimate of a UE based on its handover count $H$, we need to know the PMF $f_H(h)$ of handover count. For a scenario as in Fig.~\ref{fig:Poisson_Line_Intersect}(b), an exact expression for the mean number of handovers can be derived as \cite{Moller_PoissonVoronoi, MobilityJeff}
\begin{align}
E[H] = \frac{4vT\sqrt{\lambda}}{\pi}. \label{eq:MeanNo_HOs}
\end{align}
To the best of our knowledge, there is no expression available for the PMF of handover count in the literature. Deriving an expression for the PMF $f_H(h)$ is a complicated and laborious task, which might result into a mathematically intractable expression\cite{Moller_PoissonVoronoi}. Hence, we derive an approximation to the PMF $f_H(h)$ in this paper.

There are several papers in the literature where the PPP based parameters are approximated rather than deriving the exact expressions due to the complexity involved in the derivation of exact expressions. For example, in \cite{TanemuraPPP}, geometrical characteristics of the perimeter, area and number of edges in a 2-dimensional Voronoi cell, and volume, surface area and number of faces in a 3-dimensional Voronoi cell are approximated by fitting generalized gamma distribution to the respective histograms. Similarly, the distributions of 2-dimensional cell area and 3-dimensional cell volume are approximated in \cite{Ferenc2007518} by fitting some simple expressions.

In this paper, we derive two approximations to the PMF $f_H(h)$ for the handover count:\vspace{1mm}\\
1) approximation $f^{\rm g}_H(h)$ derived using gamma distribution;\\
2) approximation $f^{\rm n}_H(h)$ derived using Gaussian distribution.\vspace{1mm}\\
These two approximations to the handover count PMF will be discussed in more detail in Section~\ref{sec:ApproxPMF}, and their accuracies will be further investigated and compared in Section~\ref{sec:PmfApproxAccuracy}.

\section{Approximation of the Handover Count PMF Using Gamma and Gaussian Distributions}
\label{sec:ApproxPMF}
In this section, two approximations for the PMF of handover count will be introduced. In each approximation method, the parameters of a distribution (gamma distribution or Gaussian distribution) will be approximated using the curve fitting tools in Matlab.

\subsection{Approximation of the PMF of Handover Count using Gamma Distribution}
\label{sec:ApproxGamma}
Gamma distribution has been commonly used in approximating the statistical distribution of the parameters related to PPPs such as area, volume, number of edges, etc., of Poisson Voronoi cells \cite{TanemuraPPP, Ferenc2007518}. It can be effectively used to approximate the handover count PMF. The gamma PDF can be expressed using the shape parameter $\alpha > 0$ and rate parameter $\beta > 0$ as
\begin{align}
f^{\rm g}(x) = \frac{\beta^\alpha}{\Gamma(\alpha)} x^{\alpha-1} e^{-\beta x}, \mbox{ for } x\in (0, \infty), \label{eq:GammaPDF}
\end{align}
where $\Gamma(\alpha) = \int_0^\infty t^{\alpha-1} e^{-t}\mathrm{d}t$ is the gamma function. Fitting the gamma PDF to the PMF $f_H(h)$ is not a straight forward task since the gamma PDF is a continuous function while the PMF $f_H(h)$ is a discrete function. Therefore, we fit the subsets of gamma PDF to the PMF $f_H(h)$.
\begin{figure}[htp]
\center
\includegraphics[width = 2.4in]{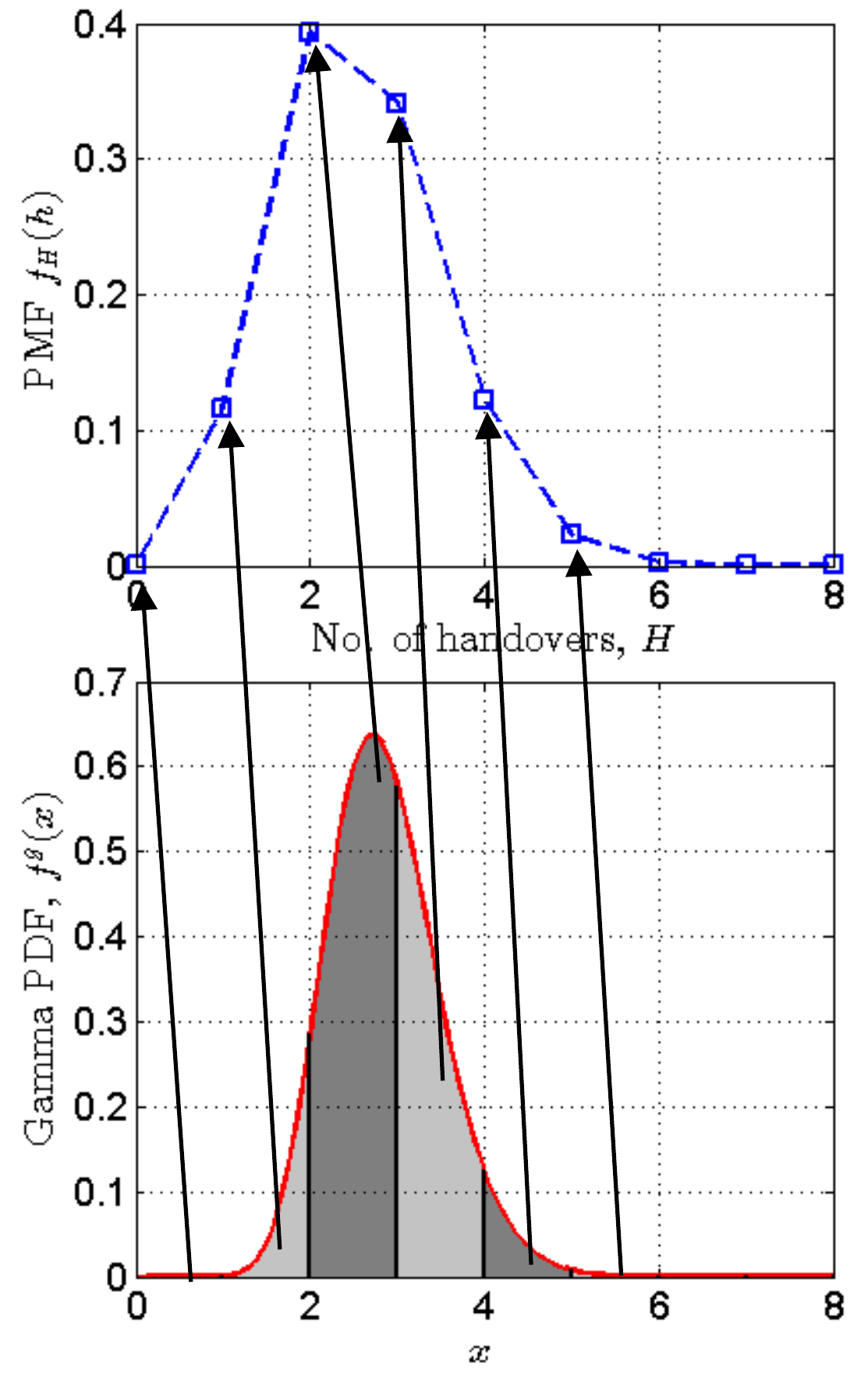}
\vspace{-2mm}
\caption{Illustration of fitting gamma distribution to the handover count PMF.}
\label{fig:FittingIllustration}
\vspace{-4mm}
\end{figure}
The subsets of gamma PDF can be obtained by integrating the gamma PDF between the integer values of $x$ as illustrated in Fig.~\ref{fig:FittingIllustration}. Integrating the gamma PDF between 0 and 1 provides the PMF value for $h=0$, integrating the gamma PDF between 1 and 2 provides the PMF value for $h=1$, and so on. This process can be mathematically described as:
\begin{align}
f^{\rm g}_H(h) = \int_h^{h+1} f^{\rm g}(x) {\rm d}x, \mbox{ for } h \in \{0,1,2,...\}, \label{eq:f_hat}
\end{align}
where, $f^{\rm g}_H(h)$ is an approximation to the PMF $f_H(h)$. Substituting \eqref{eq:GammaPDF} into \eqref{eq:f_hat}, we get
\begin{align}
f^{\rm g}_H(h) &= \int_h^{h+1} \frac{\beta^\alpha}{\Gamma(\alpha)} x^{\alpha-1} e^{-\beta x} {\rm d}x = \frac{\beta^\alpha}{\Gamma(\alpha)} \int_h^{h+1} \hspace{-2mm} x^{\alpha-1} e^{-\beta x} {\rm d}x, \mbox{ for } h \in \{0,1,2,...\}.
 \label{eq:fg_hat}
\end{align}
With the change of variable $x = t/\beta$, we can rewrite \eqref{eq:fg_hat} in its equivalent form as
\begin{align}
f^{\rm g}_H(h) &= \frac{\beta^\alpha}{\Gamma(\alpha)} \int_{\beta h}^{\beta(h+1)} \left(\frac{t}{\beta}\right)^{\alpha-1} e^{-t} \frac{{\rm d}t}{\beta} =\frac{1}{\Gamma(\alpha)} \int_{\beta h}^{\beta(h+1)} \hspace{-5mm} t^{\alpha-1} e^{-t} {\rm d}t = \frac{\Gamma\big(\alpha,\beta h, \beta(h+1)\big)}{\Gamma(\alpha)}, \label{eq:fg_hat_final}
\end{align}
where, $\Gamma\big(\alpha, \beta h, \beta(h+1)\big)=\int_{\beta h}^{\beta(h+1)} t^{\alpha-1} e^{-t}\mathrm{d}t$ is the generalized incomplete gamma function. However, for the approximation in \eqref{eq:fg_hat_final} to be accurate, the values for $\alpha$ and $\beta$ parameters should be chosen such that the mean squared error (MSE) between $f_H(h)$ and $f^{\rm g}_H(h)$ is minimized.

\begin{lemma}
The $\alpha$ and $\beta$ parameters for the approximation in \eqref{eq:fg_hat_final} which minimize the MSE between $f_H(h)$ and $f^{\rm g}_H(h)$ can be expressed as
\begin{align}
\alpha &= 2.7 + 4 d \sqrt{\lambda}, \label{eq:AlphaApprox}\\
\beta &= \pi + \frac{0.8}{0.38+ d \sqrt{\lambda}}, \label{eq:BetaApprox}
\end{align}
where $d=vT$ is the distance traveled by UE during the handover-count measurement time.
\end{lemma}
\vspace{-2mm}
\begin{proof}
See Appendix~\ref{app:AlphaBetaStudy}.
\end{proof}
\vspace{2mm}

Using \eqref{eq:fg_hat_final}-\eqref{eq:BetaApprox}, we can capture the statistical distribution of handover count if the SBS density and the distance traveled by UE are known. Through \eqref{eq:AlphaApprox} and \eqref{eq:BetaApprox}, it can also be noticed that the handover-count distribution depends on the distance traveled by UE, rather than the UE velocity or the time period independently.

\subsection{Approximation of the PMF of Handover Count using Gaussian Distribution}
\label{sec:ApproxGaussian}
In the previous section, approximation to the handover count PMF $f_H(h)$ was derived using gamma distribution which resulted into an expression in integral form. In this section, we will approximate the PMF $f_H(h)$ using Gaussian distribution which results into a closed form expression for the handover-count PMF.

The PDF of Gaussian distribution can be expressed as a function of mean $\mu$ and variance $\sigma^2$:
\begin{align}
f^{\rm n}\left(x\right) = \frac{1}{\sqrt{2\pi\sigma^2}} e^{-\frac{(x-\mu)^2}{2\sigma^2}}. \label{eq:GaussianDistribution}
\end{align}
Since the PMF $f_H(h)$ is discrete and the Gaussian distribution in \eqref{eq:GaussianDistribution} is continuous, we consider only non-negative integer samples of the Gaussian distribution for the fitting process. Henceforth, the approximation to $f_H(h)$ can be expressed as,
\begin{align}
f^{\rm n}_H(h) = \frac{1}{\sqrt{2\pi\sigma^2}} e^{-\frac{(h-\mu)^2}{2\sigma^2}}, \mbox{ for } h\in\{0,1,2,...\}. \label{eq:fn_hat_final}
\end{align}
The values of $\mu$ and $\sigma^2$ should be chosen to minimize the MSE between $f_H(h)$ and $f^{\rm n}_H(h)$.

\begin{lemma}
The $\mu$ and $\sigma^2$ parameters for the approximation in \eqref{eq:fn_hat_final} which minimize the MSE between $f_H(h)$ and $f^{\rm n}_H(h)$ can be expressed as
\begin{align}
\mu &= \frac{4 d \sqrt{\lambda}}{\pi}, \label{eq:mu}\\
\sigma^2 &= 0.07 + 0.41 d \sqrt{\lambda}, \label{eq:sigma2}
\end{align}
where $d=vT$ is the distance traveled by the UE during the handover count measurement time.
\end{lemma}
\begin{proof}
See Appendix~\ref{app:MuSigma2Study}.
\end{proof}

\section{Cramer-Rao Lower Bound for Velocity Estimation}
\label{sec:CRLB}
CRLB can be used to serve as a lower bound on the variance of an unbiased estimator \cite{EstimationTheoryBook}. An estimator is said to be unbiased if its expected value is same as the true value of the parameter being estimated. An unbiased estimator whose variance can achieve the CRLB is said to be an efficient estimator, and it can achieve minimum MSE among all the unbiased estimators. In some scenarios, it might not be feasible to determine an efficient estimator. In that case, the unbiased estimator with the smallest variance is said to be a MVU estimator.

In this section, we will derive the CRLB for velocity estimation using the mathematical tools from estimation theory. Since we have two different approximations for the PMF of the number of handovers, we will obtain two separate CRLB expressions.

\subsection{CRLB Derivation using Gamma PMF Approximation}
In this sub-section we will obtain the CRLB by considering the handover count PMF approximation that was derived using gamma distribution in Section~\ref{sec:ApproxGamma}.
\begin{theorem}
In a Poisson-Voronoi tessellation of small cells with SBS node density $\lambda$, let a UE travel with velocity $v$ over a linear trajectory and make $H$ handovers over a time duration $T$. If the PMF of the handover count can be expressed using $f_H^{\rm g}(H; v)$ as in \eqref{eq:fg_hat_final}, then the CRLB for velocity estimation is given by
\begin{align}
{\rm var}(\hat{v}) \geq \frac{1}{\mathbb{E}\left[\left(\frac{\partial \log f_H^{\rm g}(H; v)}{\partial v}\right)^2\right]}, \label{eq:CRLB}
\end{align}
where, $\mathbb{E}[\cdot]$ is the expectation operator with respect to $H$, and
\begin{align}
\frac{\partial \log f_H^{\rm g}(h; v)}{\partial v} &= \frac{4 T \sqrt{\lambda} \beta^\alpha}{\alpha^2 \ \Gamma(\alpha, \beta h, \beta(h+1))} \Big[h^\alpha {_2F_2}(\alpha,\alpha;\alpha+1,\alpha+1;-\beta h) \nonumber\\
&\hspace{1.6cm}- (h+1)^\alpha {_2F_2}(\alpha,\alpha;\alpha+1,\alpha+1;-\beta(h+1)) \Big]\nonumber\\
&\hspace{0.6cm}- \frac{4 T \sqrt{\lambda}}{\Gamma(\alpha, \beta h, \beta(h+1))}\Big[\gamma(\alpha,\beta h)\log(\beta h)-\gamma(\alpha,\beta(h+1))\log(\beta(h+1))\Big] \nonumber\\
&\hspace{0.6cm} + \frac{0.8T \sqrt{\lambda}\beta^{\alpha-1}e^{-\beta h}\left[h^\alpha - e^{-\beta} (h+1)^\alpha \right]}{\Gamma(\alpha, \beta h, \beta(h+1))(0.38+v T \sqrt{\lambda})^2} - 4T\sqrt{\lambda}\ \psi(\alpha), \label{eq:DiffLogPMF_1}
\end{align}
where, $\psi(\cdot)$ is digamma function, $\gamma(\alpha,x) = \int_0^x t^{\alpha-1}e^{-t} {\rm d}t$ is lower incomplete gamma function, and ${_2F_2}(a_1,a_2;b_1,b_2;z)$ is generalized hypergeometric function which is expressed as
\begin{align}
{_2F_2}(a_1,a_2;b_1,b_2;z) = \sum_{k=0}^\infty\frac{(a_1)_k (a_2)_k}{(b_1)_k (b_2)_k} \frac{z^k}{k!},
\end{align}
where, $(a)_0 = 1$ and $(a)_k = a(a+1)(a+2)...(a+k-1)$, for $k \geq 1$.
\end{theorem}
\begin{proof}
See Appendix~\ref{app:CRLB_Gamma}.
\end{proof}

Due to the complexity of expression in \eqref{eq:DiffLogPMF_1}, it is impractical to derive the right hand side (RHS) of \eqref{eq:CRLB} in closed form. For this reason, we can only find asymptotic CRLB by numerically evaluating the RHS of \eqref{eq:CRLB}. Through simulations, we generate $N$ samples of the random variable $H$ and denote them as $\{H_n\}$, for $n \in 1,2,...,N$. Using these $N$ samples, we can numerically evaluate the asymptotic CRLB using
\begin{align}
{\rm var}(\hat{v}) \geq \frac{N}{\sum_{m=H_{\rm min}}^{H_{\rm max}}\left(N_m\left(\frac{\partial \log f_H^{\rm g}(m; v)}{\partial v}\right)^2\right)}, \label{eq:AymptoticCRLB_1}
\end{align}
where, $H_{\rm max} = \max\{H_n: \forall n\in 1,2,...,N\}$, is the maximum value of $H_n$, $H_{\rm min} = \min\{H_n: \forall n\in 1,2,...,N\}$, is the minimum value of $H_n$, and $N_m = \sum_{n=1}^N 1\{H_n=m\}$ is the number of elements in the set $\{H_n\}$ that are equal to $m$. Here, $1\{\cdot\}$ is the indicator function whose value is $1$ if the condition inside the braces is true, $0$ otherwise.

\subsection{CRLB Derivation using Gaussian PMF Approximation}
In this sub-section, we will obtain the CRLB by considering the PMF approximation using Gaussian distribution that was derived in Section~\ref{sec:ApproxGaussian}.
\begin{theorem}
In a Poisson-Voronoi tessellation of small cells with SBS node density $\lambda$, let a UE travel with velocity $v$ over a linear trajectory and make $H$ handovers over a time duration $T$. If the PMF of the handover count can be expressed using $f_H^{\rm n}(H; v)$ as in \eqref{eq:fn_hat_final}, then the CRLB for velocity estimation is given by
\begin{align}
{\rm var}(\hat{v}) \geq \frac{1}{\left(\frac{\mu}{v\sigma}\right)^2 + \frac{1}{2}\left(\frac{0.41 T \sqrt{\lambda}}{\sigma^2}\right)^2}. \label{eq:CRLB_GaussApprox_1}
\end{align}
\end{theorem}
\begin{proof}
Consider the PMF approximation $f_H^{\rm n}(h)$ in \eqref{eq:fn_hat_final} which can be represented as a general Gaussian distribution, $H \sim \mathcal{N}\left(\mu,\sigma^2\right),$ where $\mu$ and $\sigma^2$ are given by \eqref{eq:mu} and \eqref{eq:sigma2} respectively. The Fisher information for the general Gaussian observations is given by \cite[Section~3.9]{EstimationTheoryBook}
\begin{align}
I(v) =& \left(\frac{\partial\mu}{\partial v}\right)^2 \frac{1}{\sigma^2} + \frac{1}{2 \left(\sigma^2\right)^2} \left(\frac{\partial \sigma^2}{\partial v}\right)^2 = \left(\frac{4 T \sqrt{\lambda}}{\pi}\right)^2 \frac{1}{\sigma^2} + \frac{1}{2 \left(\sigma^2\right)^2} \left(0.41 T \sqrt{\lambda}\right)^2.
\end{align}
Using inverse of the Fisher information, the CRLB can be expressed as in~\eqref{eq:CRLB_GaussApprox_1}.
\end{proof}

\subsection{Minimum Variance Unbiased Estimator for UE Velocity}
\label{sec:UnbiasedEstimator}
In Section~\ref{sec:ApproxGamma} and Section~\ref{sec:ApproxGaussian}, two CRLB expressions were derived by considering gamma and Gaussian distributions, respectively, for approximating the handover count PMF. In the case of using gamma distribution, the CRLB expression was complicated and not in closed form. On the other hand, in the case of using Gaussian distribution, the CRLB expression was relatively simple and in closed form. Hence, in this sub-section, we will consider the case with Gaussian distribution and derive an estimator $\hat{v}$ for a UE's velocity, which takes the number of handovers $H$ as the input. We will further derive the mean and the variance of this estimator and show that it is a MVU estimator.

To derive the MVU velocity estimator, we first use Neyman-Fisher factorization to find the sufficient statistic for $v$~\cite[Section~5.4]{EstimationTheoryBook}. Then, we make use of Rao-Blackwell-Lehmann-Scheffe (RBLS) theorem to find the MVUE~\cite[Section~5.5]{EstimationTheoryBook}. The Neyman-Fisher factorization theorem states that if we can factor the PMF $f^{\rm n}_H(h)$ as
\begin{align}
f^{\rm n}_H(h) = g\big(\mathcal F(h),v\big)r(h), \label{eq:NF}
\end{align}
where $g$ is a function depending on $h$ only through $\mathcal F(h)$ and $r$ is a function depending only on $h$, then $\mathcal F(h)$ is a sufficient statistic for $v$. Using \eqref{eq:fn_hat_final} and letting $\mathcal F(h) = h$, we can factor the PMF $f^{\rm n}_H(h)$ in the form of \eqref{eq:NF} as
\begin{align}
f^{\rm n}_H(h) &= \underbrace{\frac{1}{\sqrt{2\pi\sigma^2}} e^{-\frac{(\mathcal F(h)-\mu)^2}{2\sigma^2}}} \cdot \underbrace{1},\\
& \ \ \ \ \ \ \ \ \ g\big(\mathcal F(h),v\big) \ \ \ \ \ \ \ r(h) \nonumber
\end{align}
Therefore, the sufficient statistic for $v$ is $\mathcal F(h)=h$. The sufficient statistic can be used to find the MVU estimator by determining a function $s$ so that $\hat{v}=s(\mathcal F)$ is an unbiased estimator of $v$. By inspecting the relationship between the mean number of handovers $\bar{H}$ and the velocity $v$ in \eqref{eq:MeanNo_HOs}, we can formulate an estimator for $v$ as:
\begin{align}
\hat{v} = \frac{\pi H}{4 T \sqrt{\lambda}}. \label{eq:Estimator}
\end{align}
In order to evaluate whether this estimator is unbiased, the expectation of the above estimator can be derived as
\begin{align}
E[\hat{v}] &= E\left[\frac{\pi H}{4 T \sqrt{\lambda}}\right] = \frac{\pi}{4 T \sqrt{\lambda}} E[H] = \frac{\pi}{4 T \sqrt{\lambda}} \mu. \label{eq:EstimatorExpection}
\end{align}
Plugging \eqref{eq:mu} into \eqref{eq:EstimatorExpection}, we get
\begin{align}
E[\hat{v}] &= \frac{\pi}{4 T \sqrt{\lambda}} \frac{4 v T \sqrt{\lambda}}{\pi} = v.
\end{align}
Therefore, the estimator $\hat{v}$ expressed in \eqref{eq:Estimator} is unbiased. Since this estimator is derived through RBLS theorem, it is an MVU estimator. To determine whether it is an efficient estimator, we derive the variance of the MVU estimator as follows:
\begin{align}
{\rm var}(\hat{v}) =& {\rm var}\left(\frac{\pi H}{4 T \sqrt{\lambda}}\right)= \left(\frac{\pi}{4 T \sqrt{\lambda}}\right)^2 {\rm var}(H) = \left(\frac{v\sigma}{\mu}\right)^2. \label{eq:EstimatorVariance}
\end{align}
Comparing \eqref{eq:EstimatorVariance} with \eqref{eq:CRLB_GaussApprox_1}, we can notice that the variance of MVU estimator is greater than the CRLB, and hence, the derived estimator is not an efficient estimator. Nevertheless, in Section~\ref{sec:MVUEvariance}, we show that the variance of the MVU estimator is very close to the CRLB.

\section{Mobility State Detection}
\label{sec:MobStateProb}
In this section, we will perform statistical analysis of MSD, in which a UE is categorized into one of the three different mobility states: low mobility, medium mobility and high mobility, as in 3GPP LTE Release-8 specifications \cite{3gppRel8Overview, 3GPP_MSEenhancement,Speed_Differentiated_HO_2012}. We assume that the unbiased velocity estimator derived in Section~\ref{sec:UnbiasedEstimator} is used to estimate the UE velocity, and we will derive expressions for the probabilities that a UE is categorized into each of the three mobility states.

Using the estimated UE velocity $\hat{v}$ from \eqref{eq:Estimator}, the UE can be categorized into one of the three mobility states: low $(S_{\rm L})$, medium $(S_{\rm M})$, and high $(S_{\rm H})$, based on the following conditions:
\begin{align}
\mathcal{S} = \left\{ \begin{array}{lcl}
S_{\rm L} & \mbox{if} & \hat{v} \leq v_{\rm l}, \\
S_{\rm M} & \mbox{if} & v_{\rm l} < \hat{v} \leq v_{\rm u}, \\
S_{\rm H} & \mbox{if} & \hat{v} > v_{\rm u}, \\
\end{array}\right.
\end{align}
where, $\mathcal{S} \in \{S_{\rm L}, S_{\rm M}, S_{\rm H}\}$ is the detected mobility state of the UE. The thresholds $v_{\rm l}$ and $v_{\rm u}$ are the lower and upper velocity thresholds, respectively, based on which a UE is classified into one of the three mobility states.

\vspace{-3mm}
\subsection{Mobility State Probabilities}
For a given velocity $v$, we define mobility state probability as the probability that the UE is categorized into a particular state. We can define the following three mobility state probabilities:
\begin{align*}
P(\mathcal{S}=S_{\rm L}; v) \rightarrow& \mbox{Probability that the mobility state is detected as } S_{\rm L}, \mbox{ for a velocity } v;\\
P(\mathcal{S}=S_{\rm M}; v) \rightarrow& \mbox{Probability that the mobility state is detected as } S_{\rm M}, \mbox{ for a velocity } v;\\
P(\mathcal{S}=S_{\rm H}; v) \rightarrow& \mbox{Probability that the mobility state is detected as } S_{\rm H}, \mbox{ for a velocity } v.
\end{align*}
For a given velocity $v$, as the number of handovers $H$ is a random variable, the velocity $\hat{v}$ estimated using \eqref{eq:Estimator} is also a random variable. Hence, there will be false alarms and missed detections for calculating the mobility state.

Next, we derive analytic expressions for the mobility state probabilities. Using \eqref{eq:Estimator}, we can express the PMF of $\hat{v}$ as,
\begin{align}
\hspace{-3mm}f_{\hat{v}}(\nu)&= P(\hat{v}=\nu)= P\left(\frac{\pi H}{4T\sqrt{\lambda}}=\nu\right)= P\left(H=\frac{4T\sqrt{\lambda}\nu}{\pi}\right)=f_H\left(\frac{4T\sqrt{\lambda}\nu}{\pi}\right) = f_H(h), \label{eq:f_v_hat}
\end{align}
where, $h=\frac{4T\sqrt{\lambda}\nu}{\pi}$. Using the approximation of PMF $f_H(h)$ with Gaussian distribution as in \eqref{eq:fn_hat_final}, we can approximate the PMF of $\hat{v}$ as,
\begin{align}
\hspace{-2mm}f_{\hat{v}}(\nu) =& f_H(h) \nonumber\\
\approx& f^{\rm n}_H(h) = \frac{1}{\sqrt{2\pi\sigma^2}} e^{-\frac{(h-\mu)^2}{2\sigma^2}}, \mbox{ for } h\in\{0,1,2,...\}.\hspace{-2mm} \label{eq:fn_hat}
\end{align}
Now, we can express the three mobility state probabilities as
\begin{align}
P(\mathcal{S}=S_{\rm L}; v) &= P(\hat{v} \leq v_{\rm l}) \approx \sum_{h=0}^{h_{\rm l}}f^{\rm n}_H(h), \label{eq:L_StateProb}\\
P(\mathcal{S}=S_{\rm M}; v) &= P(v_{\rm l} < \hat{v} \leq v_{\rm u}) \approx \sum_{h= h_{\rm l}+1}^{h_{\rm u}}f^{\rm n}_H(h), \label{eq:M_StateProb}\\
P(\mathcal{S}=S_{\rm H}; v) &= P(\hat{v} > v_{\rm u}) \approx \sum_{h= h_{\rm u}+1}^{\infty}f^{\rm n}_H(h), \label{eq:H_StateProb}\\
\mbox{where,}\hspace{2mm} h_{\rm l}=\Big\lfloor&\frac{4T\sqrt{\lambda}v_{\rm l}}{\pi}\Big\rfloor,\hspace{2mm} \mbox{and} \hspace{2mm} h_{\rm u}=\Big\lfloor\frac{4T\sqrt{\lambda}v_{\rm u}}{\pi}\Big\rfloor, \label{eq:HOth}
\end{align}
are the optimum lower and upper handover count thresholds for MSD, respectively. Given the velocity thresholds $v_{\rm l}$ and $v_{\rm u}$, the choice of handover count thresholds has a direct impact on the probability of correctly detecting the mobility state of a UE based on its velocity~\cite{7084911}. In \eqref{eq:HOth}, we have theoretically derived the handover-count thresholds for MSD which are optimum for the given velocity thresholds. In other related works in the literature, the handover-count thresholds for MSD have been determined through simulations. In \cite{3gpp_HOF, EnhancementsOfMSE, 7084911}, the handover-count thresholds for MSD are found heuristically by considering the cumulative distribution function (CDF) plots of handover counts, for few different UE velocities. In \cite{6934879}, the handover counts are assumed to be distributed as Gaussian PDF, and the optimum handover-count thresholds are obtained by using the handover-count PDFs for few different UE velocities. However, in these prior works, the optimum handover-count thresholds are determined for some particular values of BS density and measurement time windows. Moreover, the statistical relationship between the UE velocity and the handover count is not considered. In this paper, we have derived general expressions for the optimum handover-count thresholds as a function of SBS density $\lambda$, handover-count measurement time $T$, and velocity thresholds ($v_{\rm l}$ and $v_{\rm u}$).

\vspace{-2mm}
\subsection{Probability of Detection and Probability of False Alarm}
The \emph{probability of detection} is the probability that the mobility state of a UE is detected correctly. Mathematically, it can be expressed as
\begin{align}
P_{\rm D} = \left\{ \begin{array}{lcl}
P(\mathcal{S}=S_{\rm L}; v) & \mbox{if} & v \leq v_{\rm l}, \\
P(\mathcal{S}=S_{\rm M}; v) & \mbox{if} & v_{\rm l} < v \leq v_{\rm u}, \\
P(\mathcal{S}=S_{\rm H}; v) & \mbox{if} & v > v_{\rm u}. \\
\end{array}\right. \label{eq:ProbDetect}
\end{align}
The \emph{probability of false alarm} is the probability that the mobility state is detected incorrectly, which can be mathematically expressed as $P_{\rm FA} = 1-P_{\rm D}$.

Consider an illustrative example in which the UE velocity is $v=60$~km/h, SBS density is $\lambda=1000$~SBSs/km$^2$, and handover-count measurement time is $T=12$~s. The PMFs $f_H(h)$ and $f_{\hat{v}}(\nu)$ are shown in Figs.~\ref{fig:PMF_Of_H}(a) and \ref{fig:PMF_Of_H}(b), respectively, which were obtained through Monte Carlo Simulations. In Fig.~\ref{fig:PMF_Of_H}(b), the range of $\hat{v}$ is divided into three regions $S_{\rm L}, S_{\rm M}$ and $S_{\rm H}$ that are separated from each other through the velocity thresholds $v_{\rm l}=40$~km/h and $v_{\rm u}=80$~km/h. It can be noticed in Fig.~\ref{fig:PMF_Of_H}(b) that even though the actual UE velocity is a constant $v=60$~km/h which belongs to $S_{\rm M}$ state, the estimated velocity is spread into a range of velocities. Hence, there is a small probability that the mobility state could be erroneously detected as $S_{\rm L}$ or $S_{\rm H}$ states, which is the \emph{probability of false alarm} $P_{\rm FA}$. On the other hand, majority of the times the mobility state would be correctly detected as $S_{\rm M}$, which is the \emph{probability of detection} $P_{\rm D}$.
\begin{figure}[htp]
\vspace{-3mm}
\center
\begin{subfigure}[b]{0.35\textwidth}
\includegraphics[width=\textwidth]{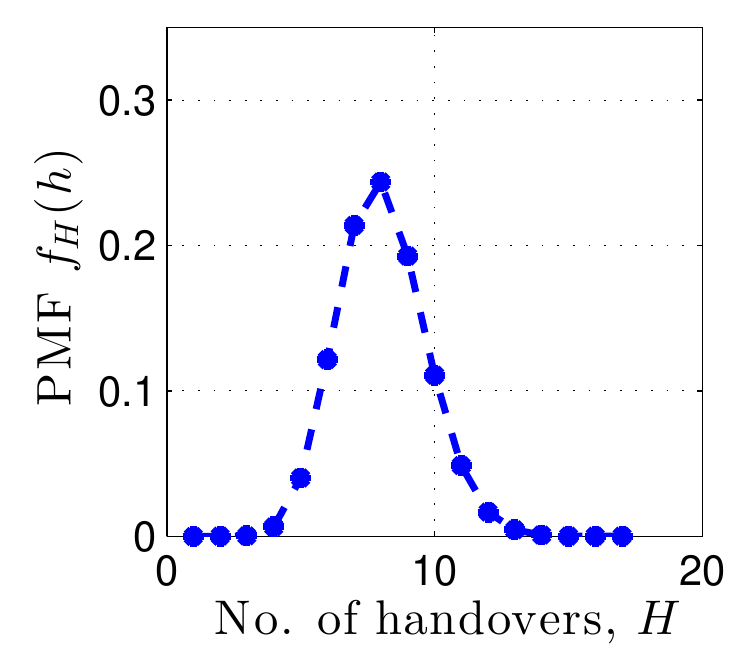}
\vspace{-7mm}\caption{}
\end{subfigure}
\begin{subfigure}[b]{0.35\textwidth}
\includegraphics[width=\textwidth]{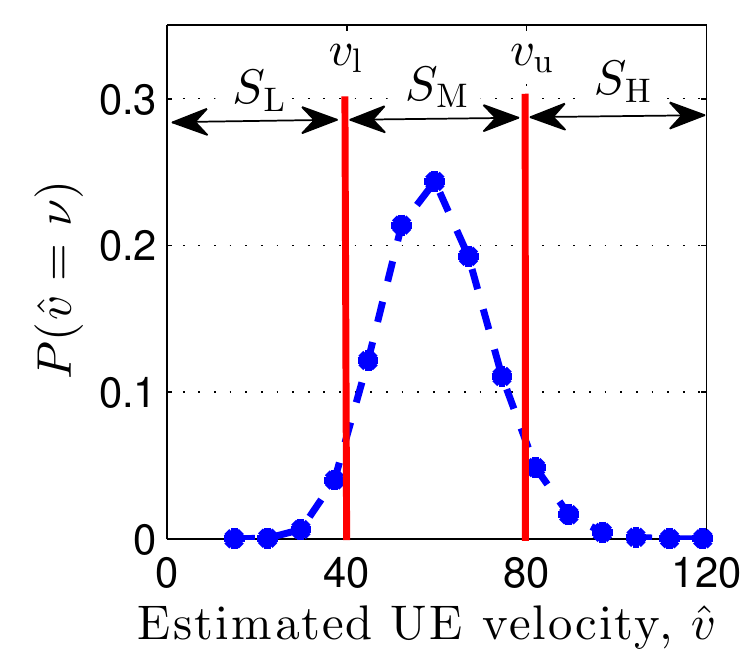}
\vspace{-7mm}\caption{}
\end{subfigure}
\vspace{-4mm}
\caption{(a) PMF of the number of handovers; (b) PMF of the estimated velocity; $v=60$~km/h, $\lambda=1000$~SBSs/km$^2$, $T=12$~s, $v_{\rm l}=40$~km/h and $v_{\rm u}=80$~km/h.}
\label{fig:PMF_Of_H}
\vspace{-1mm}
\end{figure}

With the particular parameter settings shown in Fig.~\ref{fig:PMF_Of_H}, the probabilities of different mobility states can be evaluated using \eqref{eq:L_StateProb}-\eqref{eq:H_StateProb} as $P(\mathcal{S}=S_{\rm L}; v=60) = 0.047,\ \ P(\mathcal{S}=S_{\rm M}; v=60) = 0.8821,\ \ P(\mathcal{S}=S_{\rm H}; v=60) = 0.0709,$ respectively. Then, the probabilities of detection and false alarm can be evaluated using \eqref{eq:ProbDetect} as $P_{\rm D} = P(\mathcal{S}=S_{\rm M}; v=60) = 0.8821, P_{\rm FA} = P(\mathcal{S}=S_{\rm L}; v=60) + P(\mathcal{S}=S_{\rm H}; v=60) = 0.1179,$ respectively.

\section{Numerical Results}
\label{sec:NumResults}
In this section, firstly, we will validate the accuracy of gamma PMF approximation $f^{\rm g}_H(h)$ and Gaussian PMF approximation $f^{\rm n}_H(h)$ by plotting their MSE performances. Secondly, we will plot the CRLBs and analyze the achievable accuracy of a UE's velocity estimate for different SBS density $\lambda$, UE velocity $v$, and handover count measurement time $T$. Finally, we will plot the variance of the MVU velocity estimator derived in Section~\ref{sec:UnbiasedEstimator} and show that it is approximately equal to the CRLB.

\subsection{Accuracy of PMF Approximation}
\label{sec:PmfApproxAccuracy}
Approximations to the handover count PMF were derived using gamma distribution and Gaussian distribution in Sections~\ref{sec:ApproxGamma} and \ref{sec:ApproxGaussian}, respectively. In this sub-section, we will quantify the accuracy of each approximation method by evaluating the MSE between the approximation and the PMF $f_H(h)$. The MSE can be expressed as
\begin{align}
MSE = \frac{1}{N_{\rm h}}\sum_{h=1}^{N_{\rm h}} \Big[f^u_H(h) - f_H(h)\Big]^2, \mbox{ for } u \in \{\rm g,n\}, \label{eq:MSE}
\end{align}
where $N_{\rm h}$ is the number of points in the PMF.
\begin{figure}[t]
\vspace{-3mm}
\centering
\begin{subfigure}[b]{0.49\textwidth}
\includegraphics[width=\textwidth]{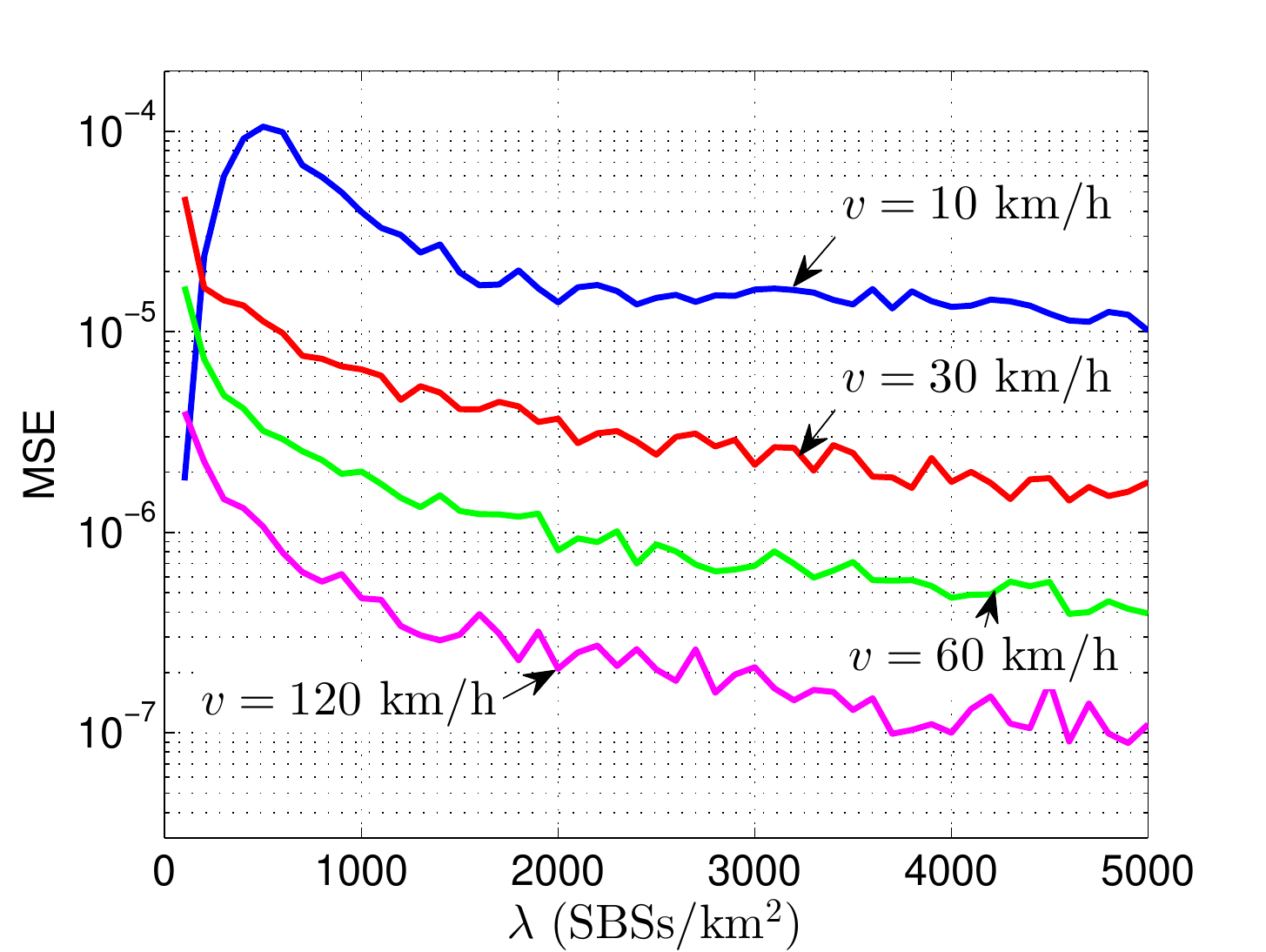}
\vspace{-6mm}\caption{}
\end{subfigure}
\begin{subfigure}[b]{0.49\textwidth}
\includegraphics[width=\textwidth]{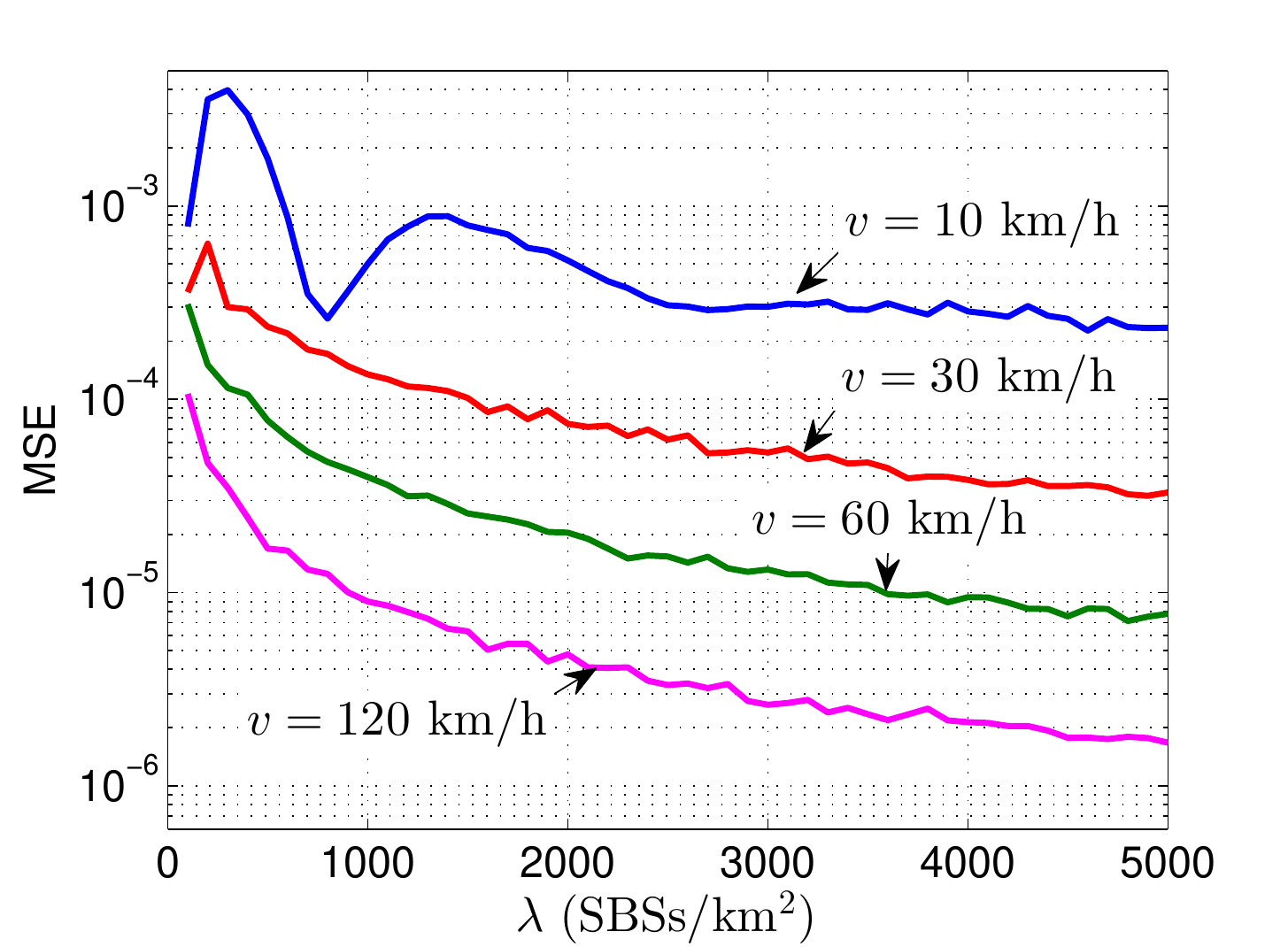}
\vspace{-6mm}\caption{}
\end{subfigure}
\vspace{-2mm}
\caption{MSE versus $\lambda$ for different $v$, and $T=12$~s; (a) with PMF approximation using gamma distribution; (b) with PMF approximation using Gaussian distribution.}
\label{fig:MSE_diff_v}
\vspace{-6mm}
\end{figure}
The characteristics of MSE with respect to the variations in $\lambda$ and $v$ are shown in Fig.~\ref{fig:MSE_diff_v} for the two approximation methods. In general, the MSEs of both approximation methods decrease with the increase in SBS density $\lambda$ or UE velocity $v$. In other words, higher SBS density and higher UE velocity leads to better accuracy of the PMF approximation. By comparing Fig.~\ref{fig:MSE_diff_v}(a) and Fig.~\ref{fig:MSE_diff_v}(b), it can be noticed that the approximation using gamma distribution provides approximately ten times smaller MSE than the approximation using Gaussian distribution. However, the approximation using gamma distribution is not in closed form and hence it is more complicated than the approximation using Gaussian distribution. Therefore, there exists a trade-off between the accuracy and the complexity while making a choice between the two approximation approaches.

\begin{figure}[t]
\vspace{-3mm}
\centering
\begin{subfigure}[b]{0.49\textwidth}
\includegraphics[width=\textwidth]{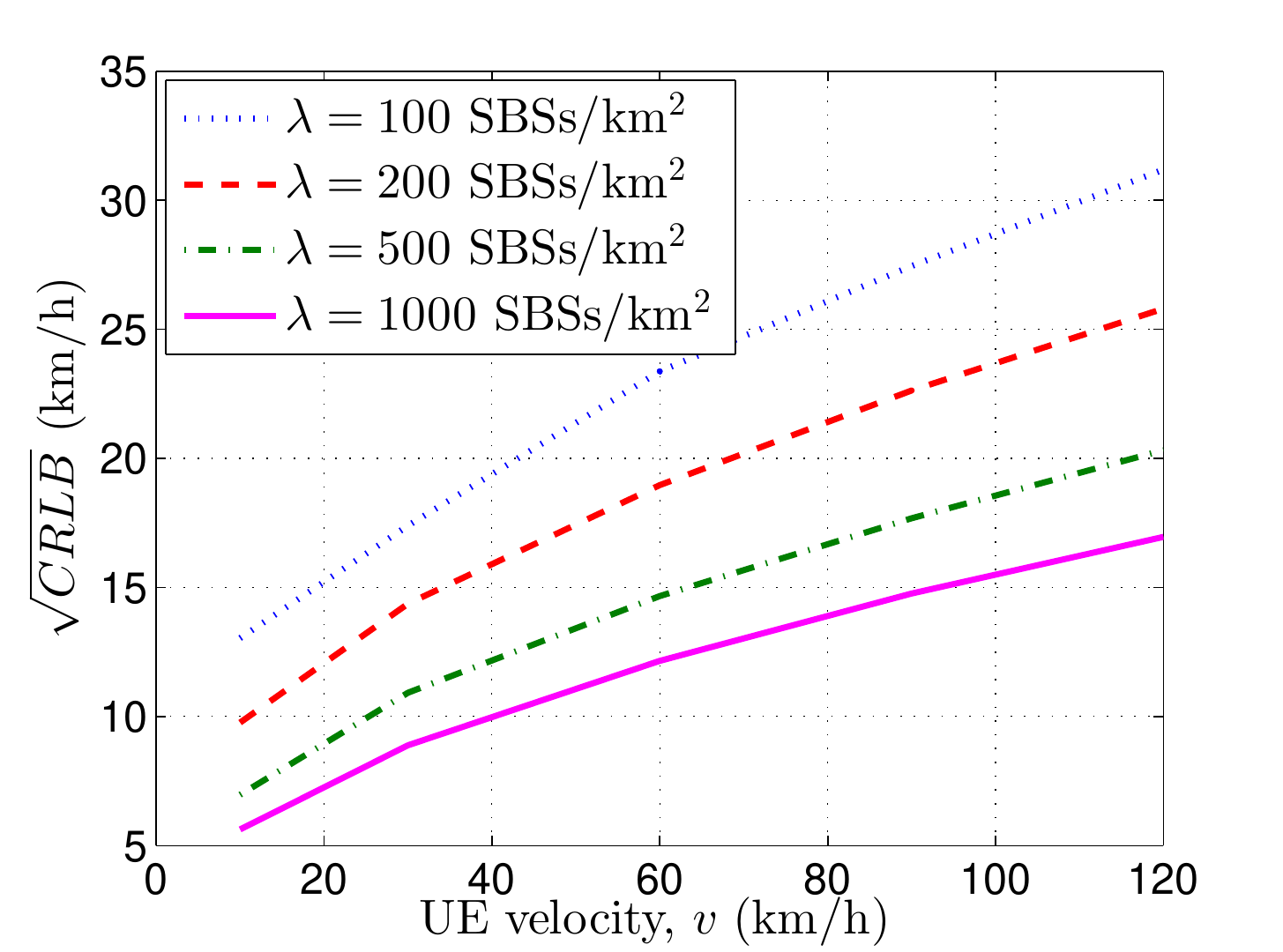}
\vspace{-6mm}\caption{}
\end{subfigure}
\begin{subfigure}[b]{0.49\textwidth}
\includegraphics[width=\textwidth]{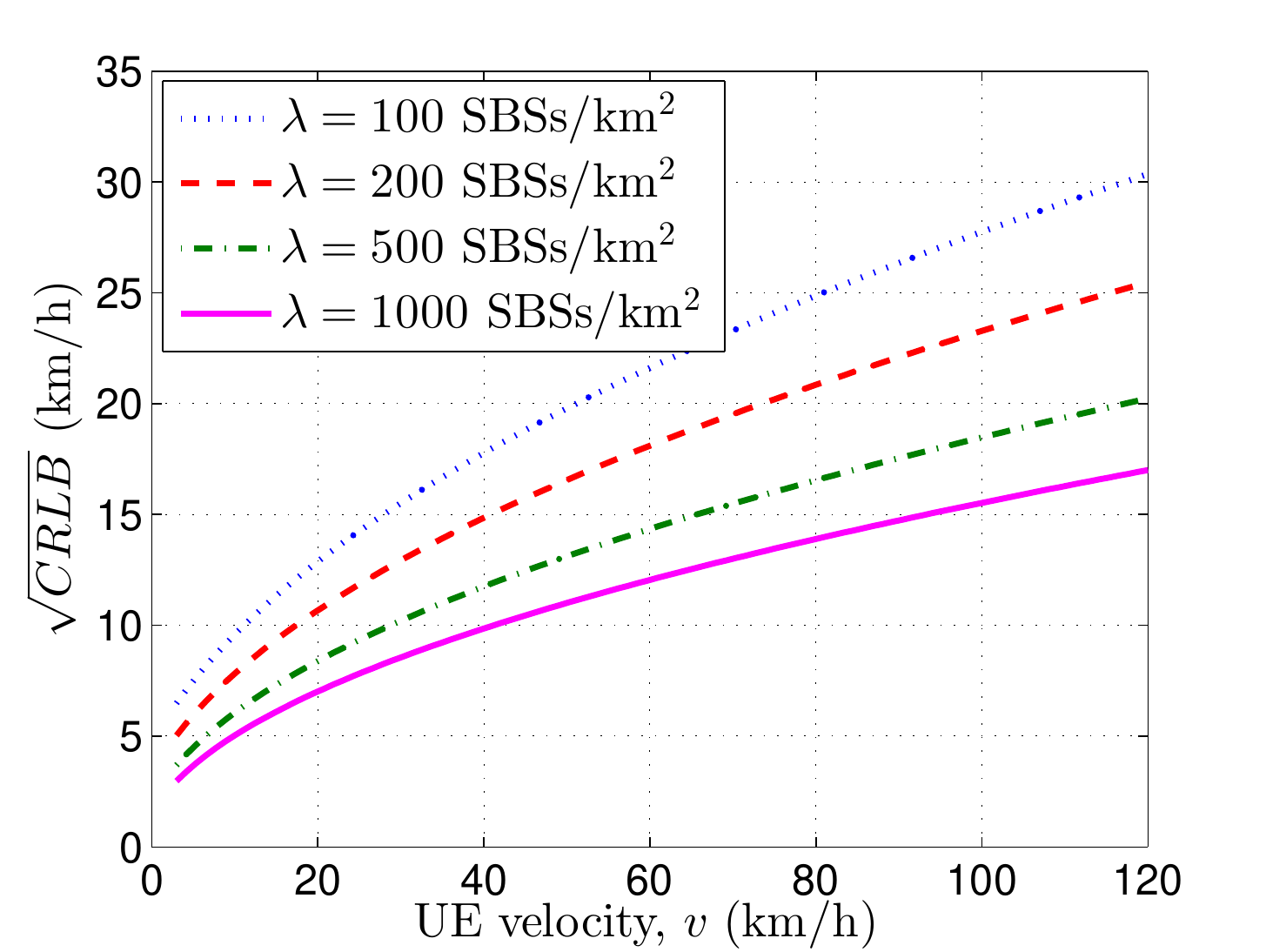}
\vspace{-6mm}\caption{}
\end{subfigure}
\vspace{-2mm}
\caption{CRLB versus $v$ for different $\lambda$, and $T=12$~s; (a) CRLB derived using $f^{\rm g}_H(h)$; (b) CRLB derived using $f^{\rm n}_H(h)$.}
\label{fig:CRLB_vs_Velocity_Diff_lambda}
\end{figure}
\subsection{CRLB Results}
The CRLB plots for UE velocity estimate $\hat{v}$ are shown in Fig.~\ref{fig:CRLB_vs_Velocity_Diff_lambda} for the variations in UE velocity $v$ and SBS density $\lambda$. In the case of gamma approximation of PMF $f_H(h)$, the CRLB plots in Fig.~\ref{fig:CRLB_vs_Velocity_Diff_lambda}(a) are obtained by numerically evaluating the expression in \eqref{eq:AymptoticCRLB_1}. On the other hand, the CRLB plots in Fig.~\ref{fig:CRLB_vs_Velocity_Diff_lambda}(b) are obtained using the closed form expression in \eqref{eq:CRLB_GaussApprox_1}, for the case of Gaussian approximation of PMF $f_H(h)$. It can be observed that the CRLB plots in Fig.~\ref{fig:CRLB_vs_Velocity_Diff_lambda}(a) and Fig.~\ref{fig:CRLB_vs_Velocity_Diff_lambda}(b) are similar and follow same trends with respect to $\lambda$ and $v$. However, we can expect that the plots in Fig.~\ref{fig:CRLB_vs_Velocity_Diff_lambda}(a) are more accurate because of the smaller MSE of gamma approximation method. On the other hand, the closed form CRLB expression in \eqref{eq:CRLB_GaussApprox_1} for the case of Gaussian approximation can provide more insights, for example, on the dependence of CRLB into different parameters.

From Fig.~\ref{fig:CRLB_vs_Velocity_Diff_lambda}, it can be noticed that the CRLB increases with the increasing UE velocity $v$. This is because the variance of the number of handovers $H$ increases with increasing $v$, which can be observed in Fig.~\ref{fig:NoHo_PMF}. In contrast, the CRLB decreases with increasing SBS density $\lambda$, which can also be intuitively understood from Fig.~\ref{fig:NoHo_PMF}. With $\lambda=100$~SBSs/km$^2$, the peaks of the PMFs for different UE velocities are closely spaced with each other making it difficult to distinguish between the different UE velocities, resulting into higher CRLB. With $\lambda=1000$~SBSs/km$^2$, the peaks of the PMFs have more spacing between them making it easier to distinguish between the different UE velocities, resulting into lower CRLB. Results in Fig.~\ref{fig:CRLB_vs_Velocity_Diff_lambda} show that for all considered SBS densities, using a handover count based approach, a UE's velocity can be estimated with a root mean squared error (RMSE) of 20 km/h for UE velocities less than 40 km/h. For UE velocities on the order of 120 km/h, velocity can still be estimated with a RMSE less than 31~km/h for densities of 100~SBSs/km$^2$, which can be further improved for higher SBS densities.
\begin{figure}[t]
\vspace{-3mm}
\centering
\begin{subfigure}[b]{0.49\textwidth}
\includegraphics[width=\textwidth]{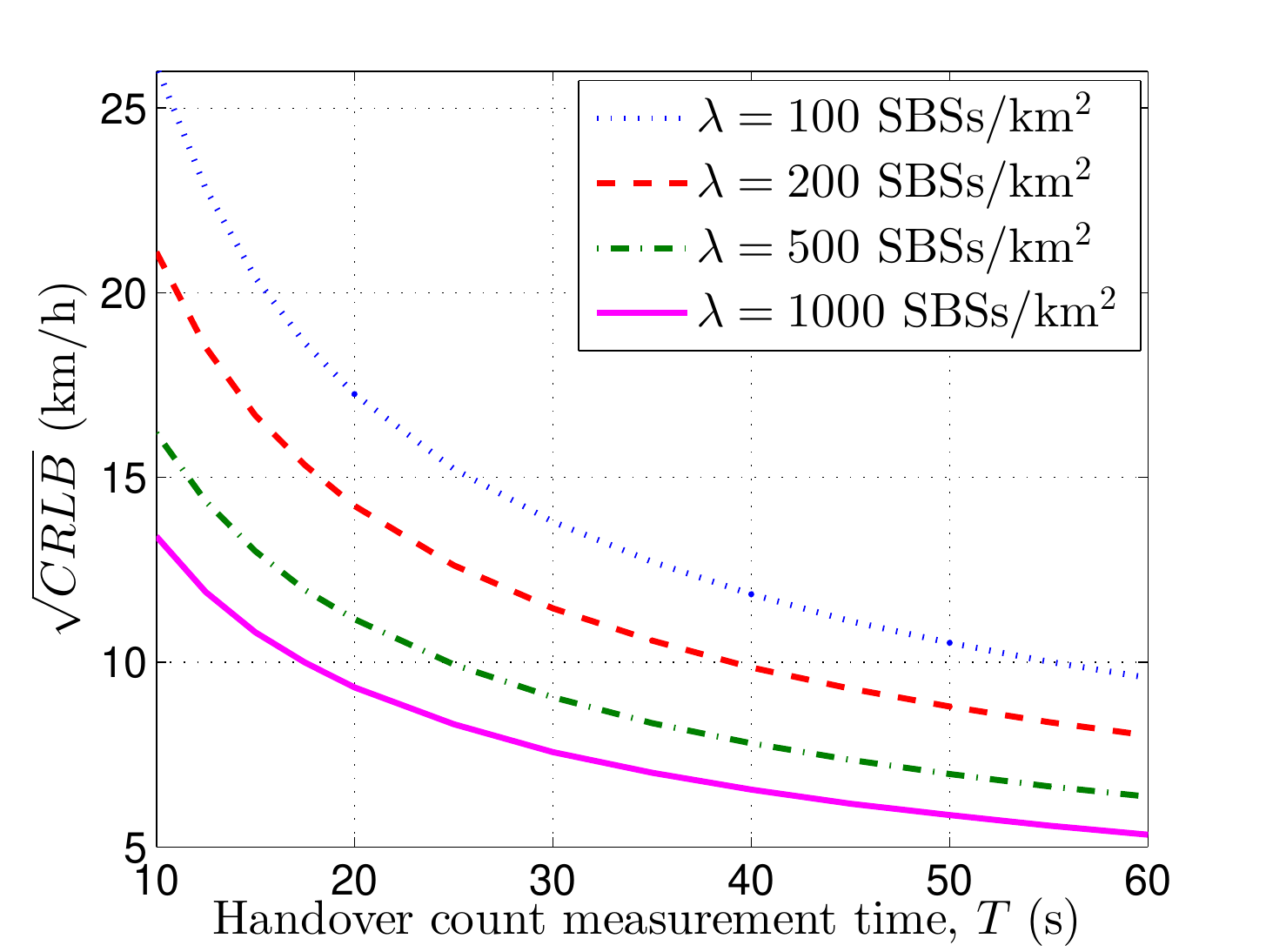}
\vspace{-6mm}\caption{}
\end{subfigure}
\begin{subfigure}[b]{0.49\textwidth}
\includegraphics[width=\textwidth]{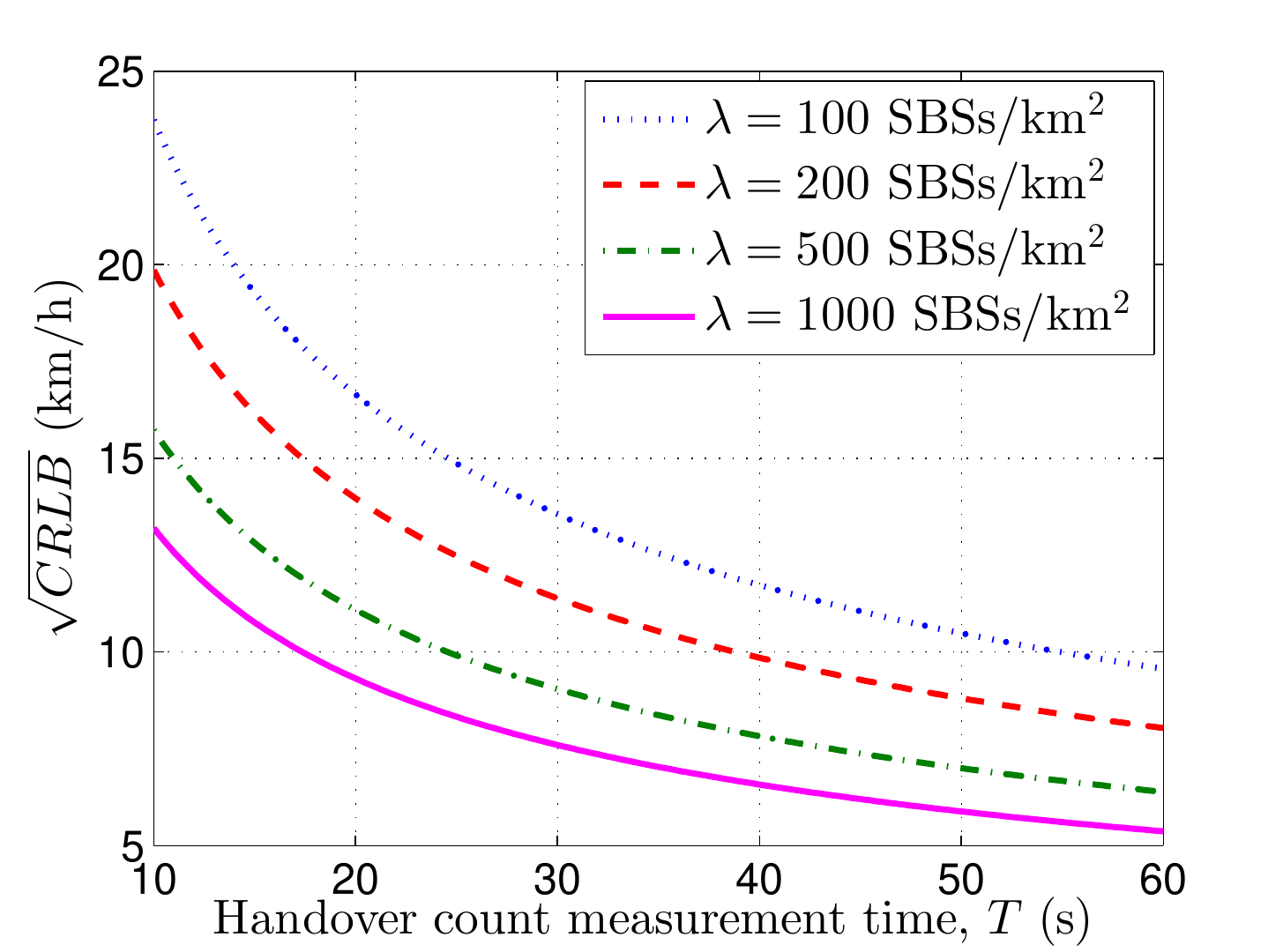}
\vspace{-6mm}\caption{}
\end{subfigure}
\vspace{-2mm}
\caption{CRLB versus $T$ for different $\lambda$ and $v=60$~km/h; (a) CRLB derived using PMF approximation $f^{\rm g}_H(h)$; (b) CRLB derived using PMF approximation $f^{\rm n}_H(h)$.}
\label{fig:CRLB_vs_Time_Diff_lambda}
\vspace{-3mm}
\end{figure}

In related works of the literature \cite{6934879,EnhancementsOfMSE,7084911,Speed_Differentiated_HO_2012}, a fixed measurement time interval of 30~s, 60~s, or 120~s is used for the MSD of UEs. However, we have used a smaller measurement time interval of $T=12$~s in Fig.~\ref{fig:CRLB_vs_Velocity_Diff_lambda} so that the velocity estimator can provide quicker results. At the same time, $T=12$~s also provides a reasonable estimation accuracy. For example, with SBS density $\lambda=500$~SBSs/km$^2$ and a UE velocity $v=120$~km/h, the CRLB of velocity estimation is just $20$~km/h as shown in Fig.~\ref{fig:CRLB_vs_Velocity_Diff_lambda}. After all, it is a choice of the service provider to use a measurement time interval based on the requirements. Therefore, the effects of variations in $T$ are also investigated in the remainder of this paper.

The effect of handover-count measurement time $T$ on the CRLB can be seen in Fig~\ref{fig:CRLB_vs_Time_Diff_lambda}. For a given UE velocity and SBS density, CRLB decreases as $T$ increases. Therefore, longer handover-count measurement time results into better accuracy of velocity estimation, with the assumption that the UE will continue traveling on a linear trajectory.

\subsection{Variance of the MVU Velocity Estimator}
\label{sec:MVUEvariance}
The variance of the MVU velocity estimator given in \eqref{eq:EstimatorVariance} is plotted in Fig.~\ref{fig:CRLB_and_EstimatorVariance}. Even though the derived MVU estimator is not an efficient estimator, it can be observed from the plots that the variance of the MVU estimator tightly matches with the CRLB. Therefore, the MVU estimator performs very close to an efficient estimator.
\begin{figure}[htp]
\centering
\begin{subfigure}[b]{0.49\textwidth}
\includegraphics[width=\textwidth]{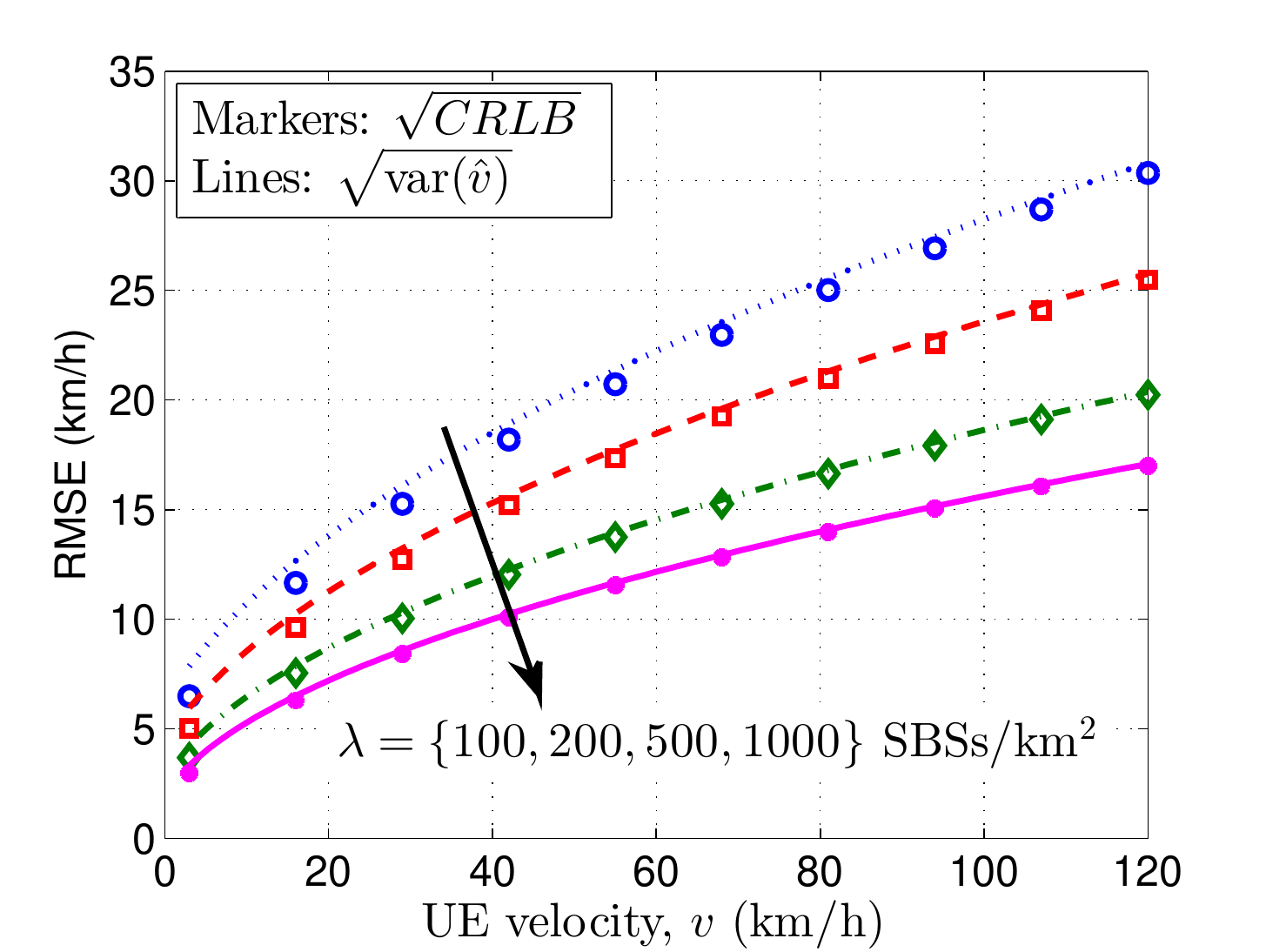}
\vspace{-6mm}\caption{}
\end{subfigure}
\begin{subfigure}[b]{0.49\textwidth}
\includegraphics[width=\textwidth]{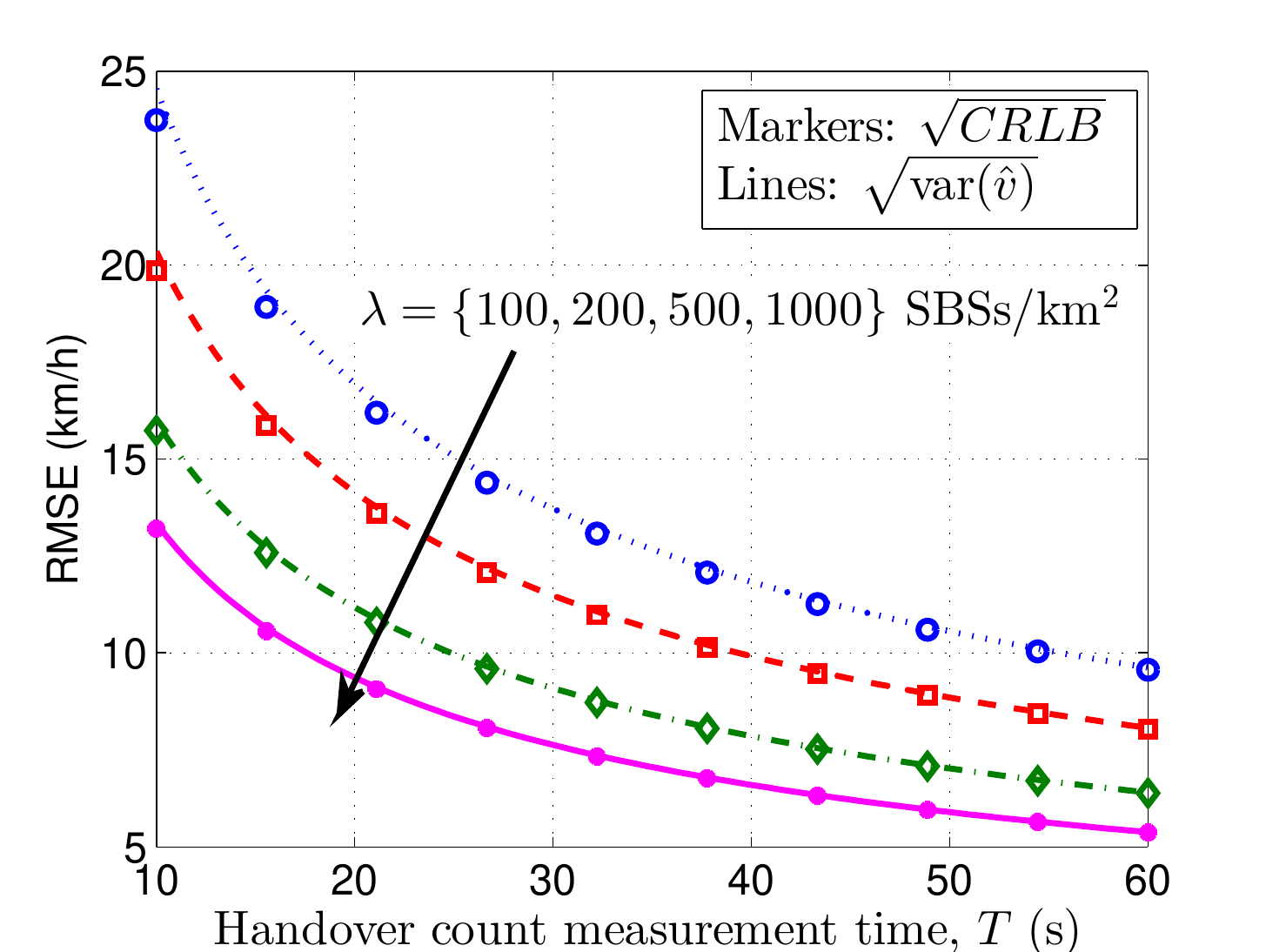}
\vspace{-6mm}\caption{}
\end{subfigure}
\vspace{-2mm}
\caption{Comparison between the CRLB and the variance of the MVU estimator; (a) for variations in $v$ and $\lambda$, with $T=12$~s; (b) for variations in $T$ and $\lambda$, with $v=60$~km/h.}
\label{fig:CRLB_and_EstimatorVariance}
\vspace{-5mm}
\end{figure}

\subsection{Mobility State Probabilities, and Probabilities of Detection and False Alarm}
\begin{figure}[h]
\centering
\begin{subfigure}[b]{0.49\textwidth}
\includegraphics[width=\textwidth]{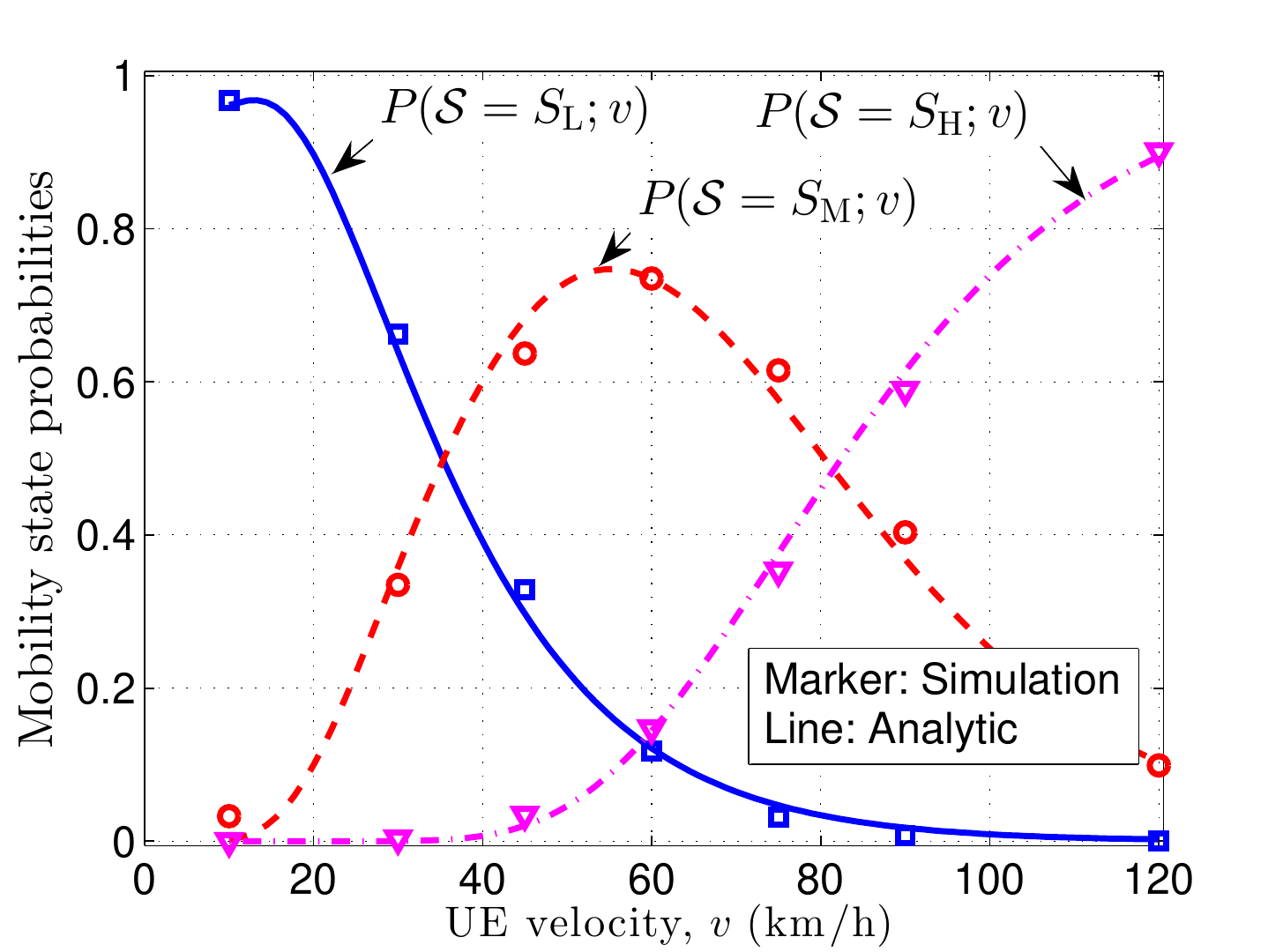}
\vspace{-7mm}\caption{}
\end{subfigure}
\begin{subfigure}[b]{0.49\textwidth}
\includegraphics[width=\textwidth]{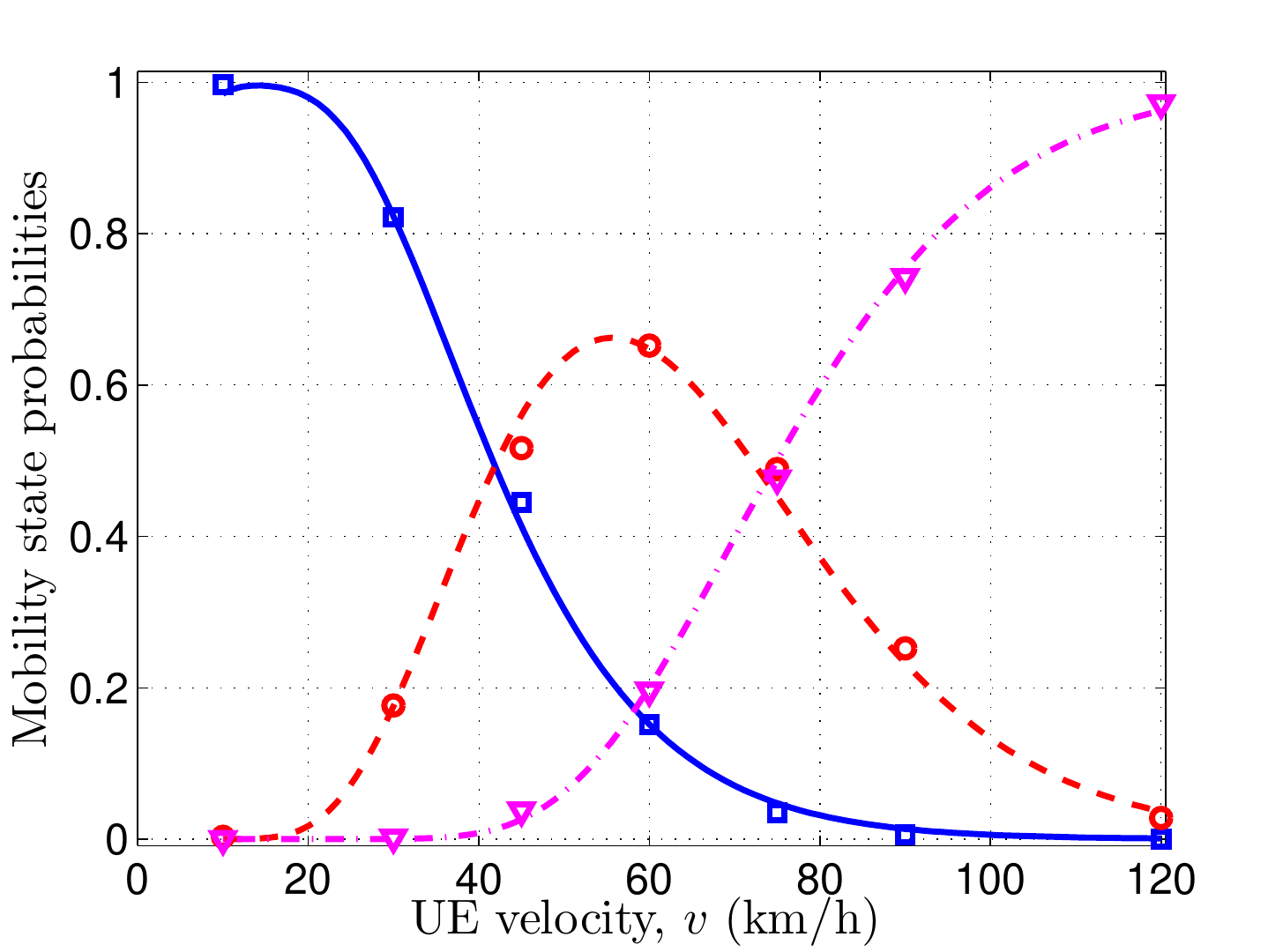}
\vspace{-7mm}\caption{}
\end{subfigure}
\begin{subfigure}[b]{0.49\textwidth}
\includegraphics[width=\textwidth]{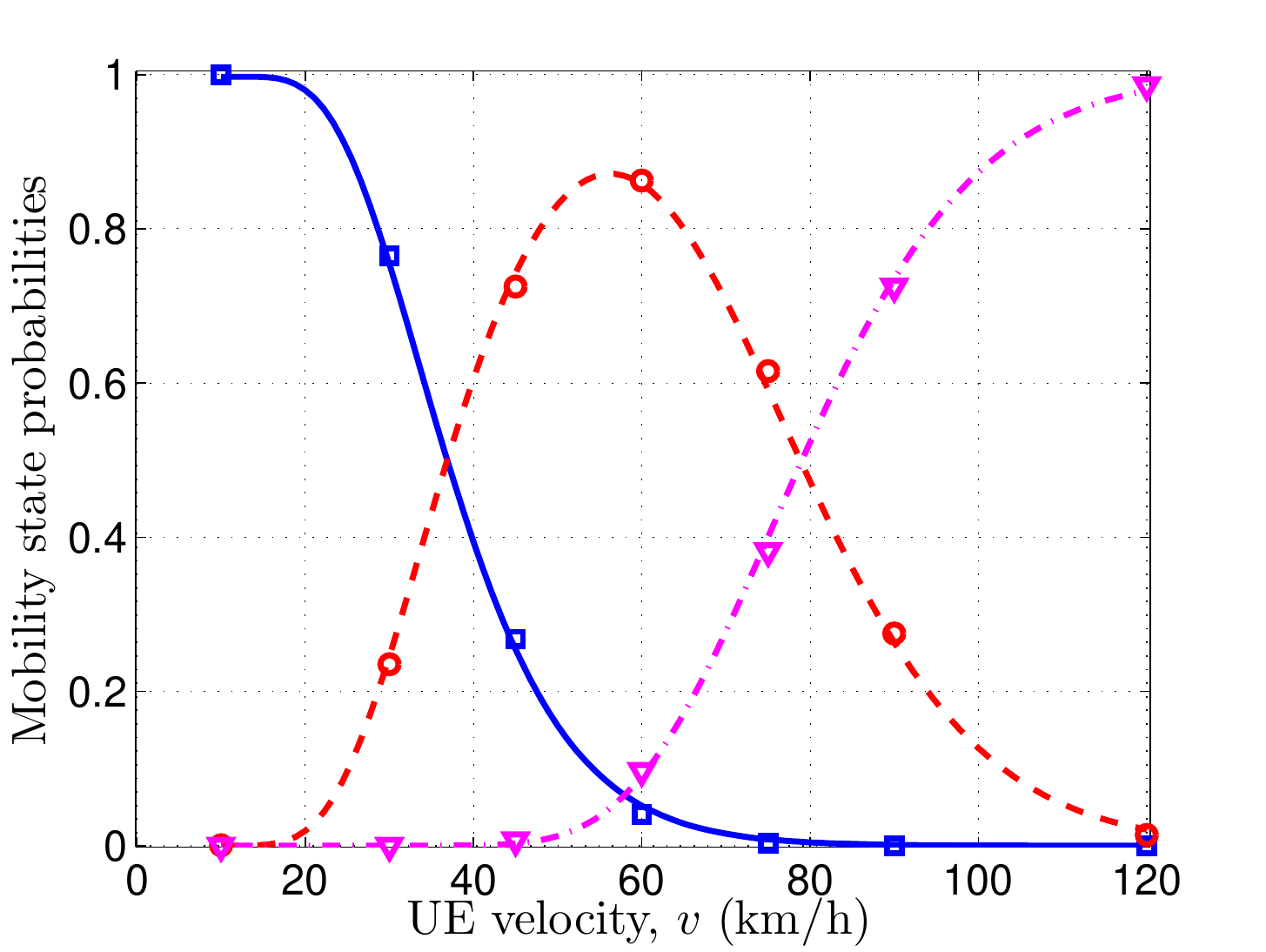}
\vspace{-7mm}\caption{}
\end{subfigure}
\begin{subfigure}[b]{0.49\textwidth}
\includegraphics[width=\textwidth]{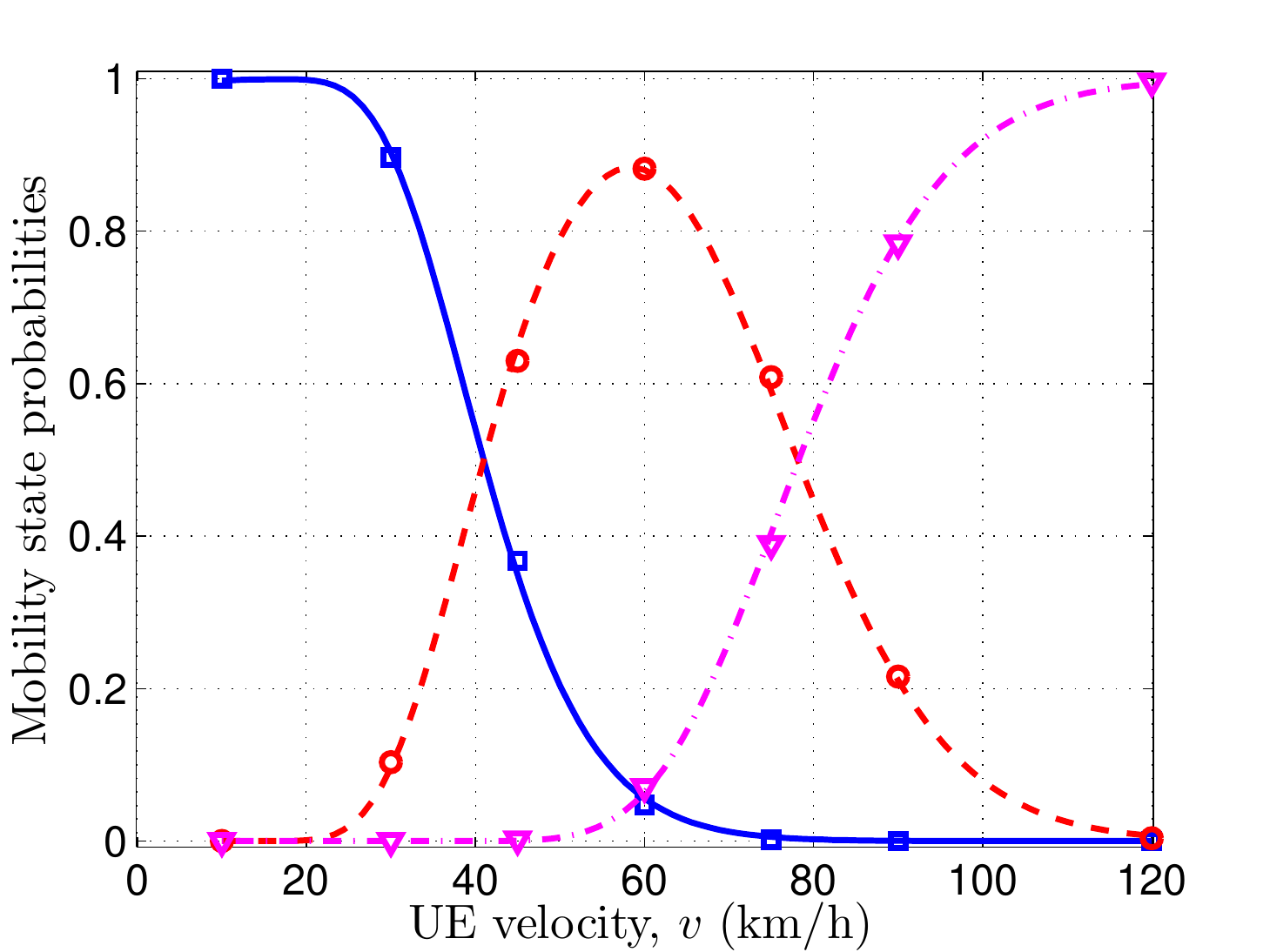}
\vspace{-7mm}\caption{}
\end{subfigure}
\vspace{-2mm}
\caption{Mobility state probabilities versus UE velocity; $T = 12$~s, $v_{\rm l}=40$~km/h and $v_{\rm u}=80$~km/h; (a) $\lambda = 100$~SBSs/km$^2$; (b) $\lambda = 200$~SBSs/km$^2$; (c) $\lambda = 500$~SBSs/km$^2$; (d) $\lambda = 1000$~SBSs/km$^2$.}
\label{fig:MobStateProb}
\vspace{-3mm}
\end{figure}
In order to classify a UE into one of the three mobility states, the velocity thresholds $v_{\rm l}$ and $v_{\rm u}$ can be set depending on the requirements of the service provider. In this paper, we have set $v_{\rm l}=40$~km/h and $v_{\rm u}=80$~km/h. The plots of mobility state probabilities versus UE velocity $v$, for four different $\lambda$ values are shown in Fig.~\ref{fig:MobStateProb}. Here, the analytic plots that were obtained using \eqref{eq:L_StateProb}-\eqref{eq:H_StateProb} are observed to match accurately with the simulation results. When we examine these plots carefully, it can be noticed that as the SBS density $\lambda$ increases, the slope of the curves during their transitions also increases, which is closer to the ideal case where the slope is equal to infinity. In other words, accuracy of MSD improves significantly with larger SBS densities.

The probabilities of detection and false alarm versus the UE velocity $v$ are shown in Fig.~\ref{fig:ProbDetectFalseAlarm} for different $\lambda$ values. As the UE's velocity nears to any one of the velocity thresholds, the probability of detection decreases and the probability of false alarm increases. The probability of detection is close to $0.5$ when the UE velocity is equal to $v_{\rm l}$ or $v_{\rm h}$.
\begin{figure}[t]
\vspace{-1mm}
\centering
\begin{subfigure}[b]{0.49\textwidth}
\includegraphics[width=\textwidth]{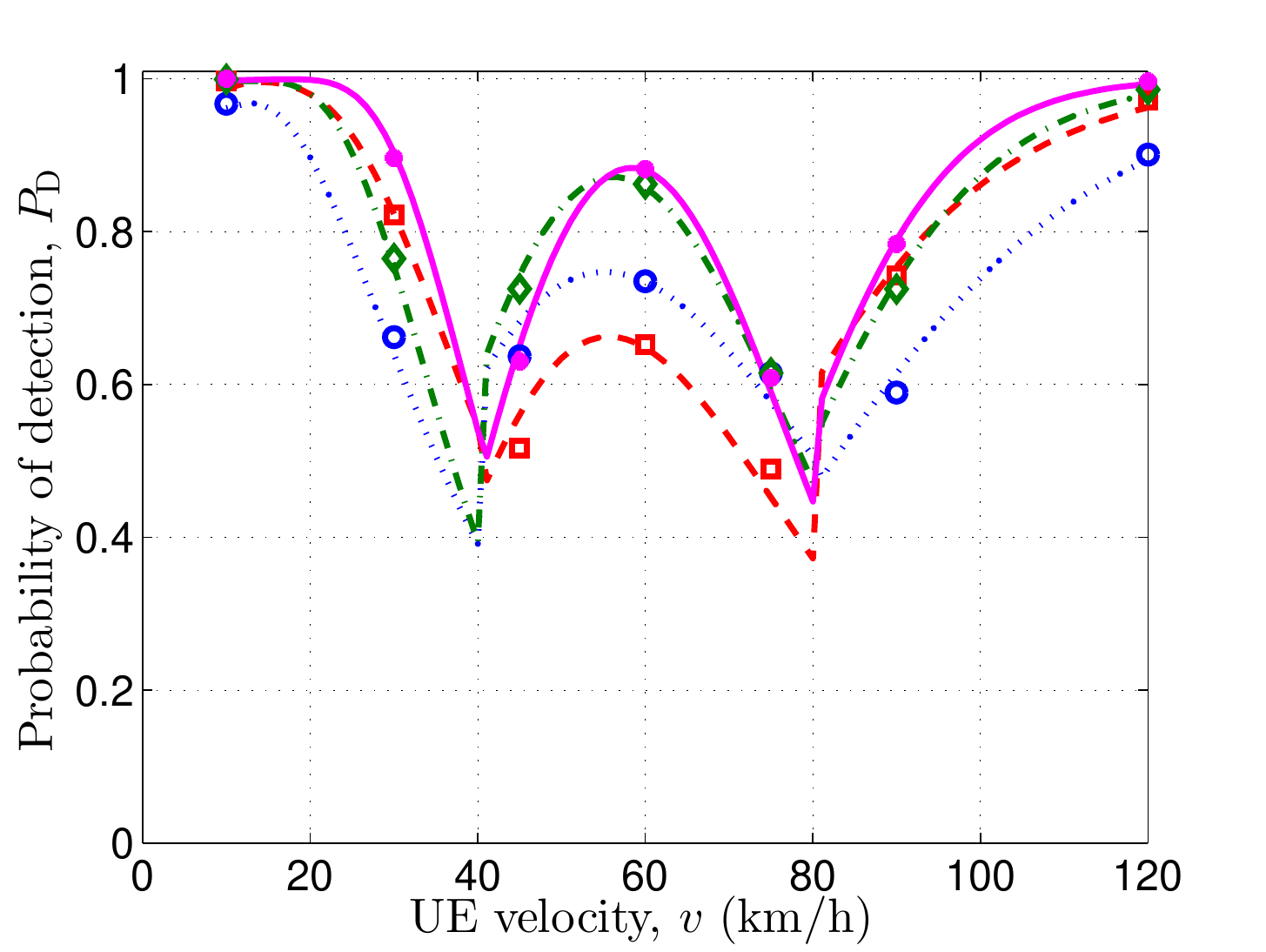}
\vspace{-7mm}\caption{}
\end{subfigure}
\begin{subfigure}[b]{0.49\textwidth}
\includegraphics[width=\textwidth]{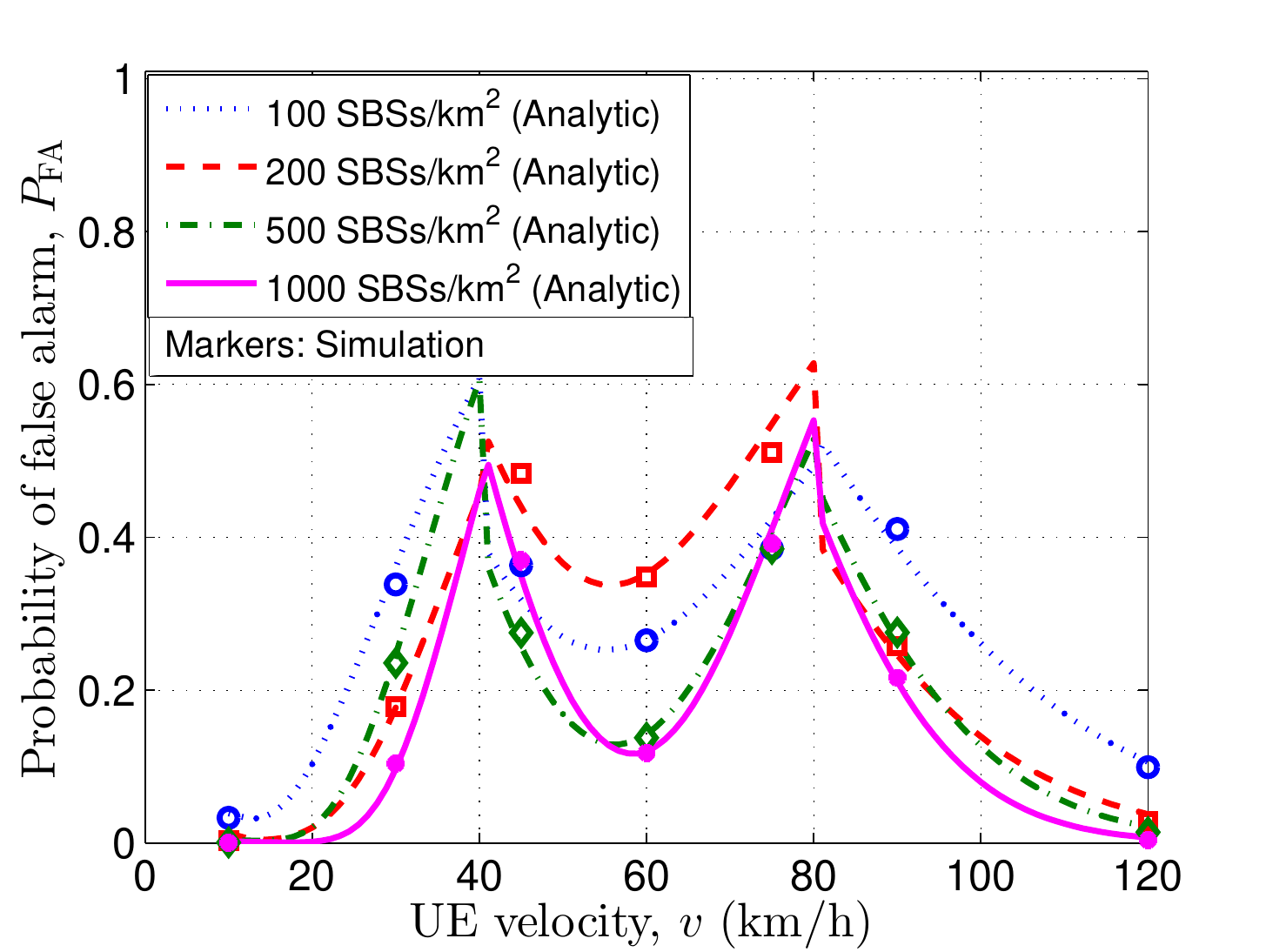}
\vspace{-7mm}\caption{}
\end{subfigure}
\vspace{-2mm}
\caption{(a) Probability of detection; (b) probability of false alarm.}
\label{fig:ProbDetectFalseAlarm}
\vspace{-3mm}
\end{figure}

The average probability of detection is plotted in Fig.~\ref{fig:Optimum_MSE_Thresholds} for different combinations of the handover count thresholds $h_{\rm l}$ and $h_{\rm u}$ such that $h_{\rm u} > h_{\rm l}$. Here, the probability of detection is averaged over UE velocity in the range of $10$~km/h to $120$~km/h, with $\lambda=500$~SBSs/km$^2$. It can be observed in Fig.~\ref{fig:Optimum_MSE_Thresholds} that the average probability of detection maximizes to $0.797$ with $h_{\rm l}=3$ and $h_{\rm u}=7$. On the other hand, computing \eqref{eq:HOth} with $\lambda=500$~SBSs/km$^2$, $T=12$~s, $v_{\rm l}=40$~km/h and $v_{\rm u}=80$~km/h also results into $h_{\rm l}=3$ and $h_{\rm u}=7$. Therefore, equations in \eqref{eq:HOth} obtained as a result of our analysis provide analytic expressions for optimum $h_{\rm l}$ and $h_{\rm u}$ that maximizes the probability of detection.
\begin{figure}[htp]
\vspace{-2mm}
\centering
\includegraphics[width=0.49\textwidth]{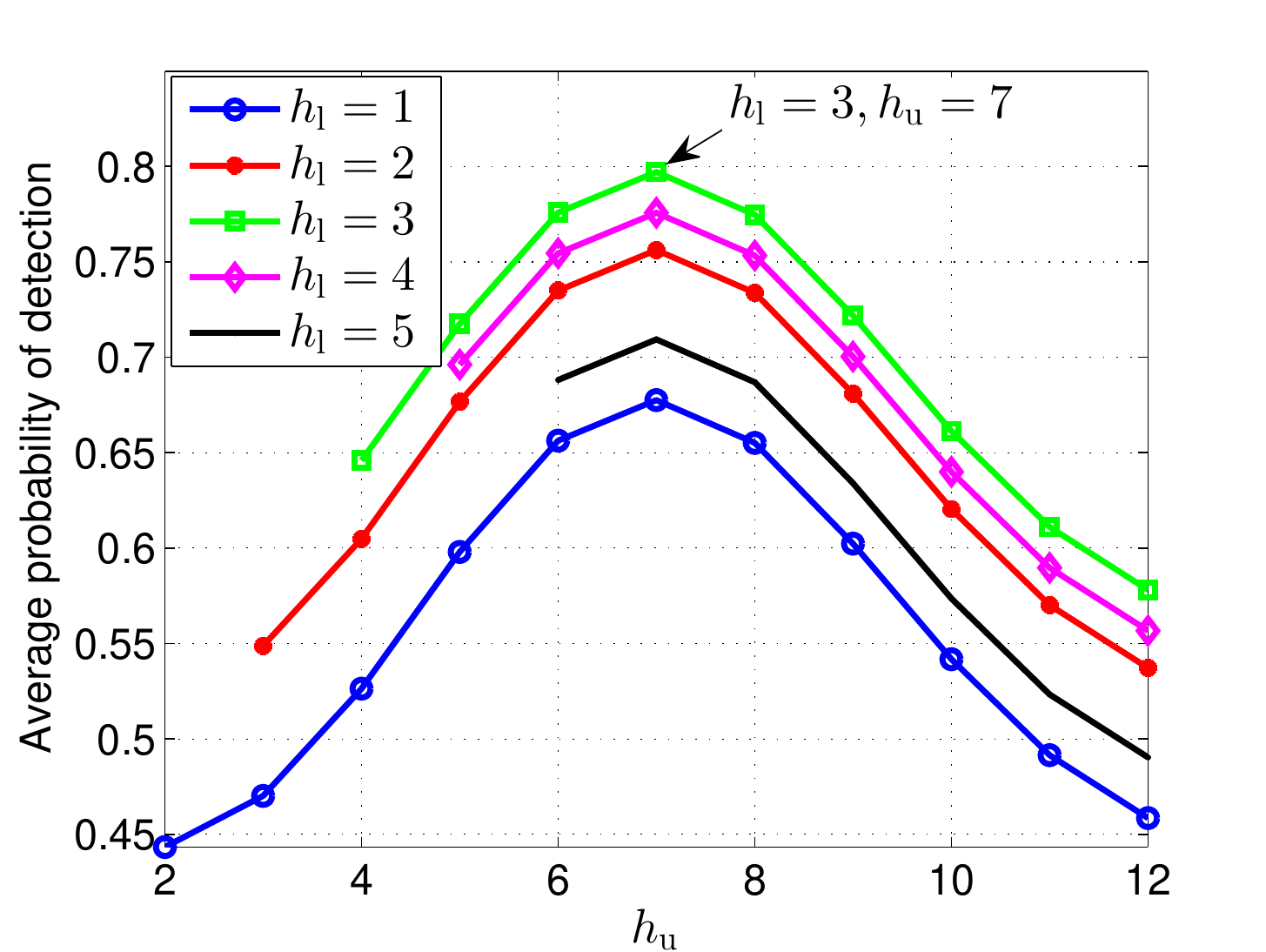}
\vspace{-2mm}
\caption{Average probability of detection for different combinations of $h_{\rm l}$ and $h_{\rm u}$, with $\lambda=500$~SBSs/km$^2$.}
\label{fig:Optimum_MSE_Thresholds}
\vspace{-4mm}
\end{figure}

\subsection{Estimator Performance with RWP Mobility Model}
In this sub-section, we assess the accuracy of velocity estimator with RWP mobility model for the UE~\cite{MobilityJeff}. In the RWP model, the UE movement trace is assumed to have a sequence of quadruples which can be defined as $\{(\mathbi{X}_{n-1},\mathbi{X}_{n},V_n,S_n)\}_{n\in\mathbb{N}}$, where $n$ denotes the $n$-th movement period, $\mathbi{X}_{n-1}$ and $\mathbi{X}_n$ denote the starting and target waypoints, respectively, during the $n$-th movement period, $V_n$ denotes the velocity, and $S_n$ denotes the pause time at the waypoint $\mathbi{X}_n$. The angle between two consecutive waypoints is uniformly randomly distributed on $[-\pi,\pi]$, while the transition length $L_n = \Vert\mathbi{X}_n - \mathbi{X}_{n-1}\Vert$ between two consecutive waypoints is i.i.d. and Rayleigh distributed with the CDF, $P(L \leq l) = 1-\exp(-\xi\pi l^2), l \geq 0,$ where $\xi$ is defined as the mobility parameter. Larger $\xi$ statistically implies that the transition lengths $L$ are shorter and may be appropriate for mobile users walking. In contrast, smaller $\xi$ statistically implies longer transitions lengths which may be appropriate for driving users.

We performed simulations for a special case of the RWP mobility model in which the UE velocity $V_n \equiv v$ is a positive constant, and the pause times $S_n=0$. The characteristics of the RMSE of velocity estimator are shown in Fig.~\ref{fig:RMSE_RWPmobility}. It can be observed that the RMSE increases with increasing $v$, and decreases with increasing $T$, similar to the characteristics of the linear mobility model. The RMSE increases with the increasing mobility parameter $\xi$. For large $\xi$, the UE switches its directions more number of times within the handover count measurement time interval, leading to larger estimation errors. On the other hand, for smaller $\xi$, the UE's direction switch rate is smaller which results into smaller RMSE. As $\xi \rightarrow 0$, the RMSE of RWP mobility model converges to the RMSE of the linear mobility model. This is because the direction switch rate tends to zero, and the UE follows a straight line trajectory indefinitely.
\begin{figure}[t]
\centering
\begin{subfigure}[b]{0.49\textwidth}
\includegraphics[width=\textwidth]{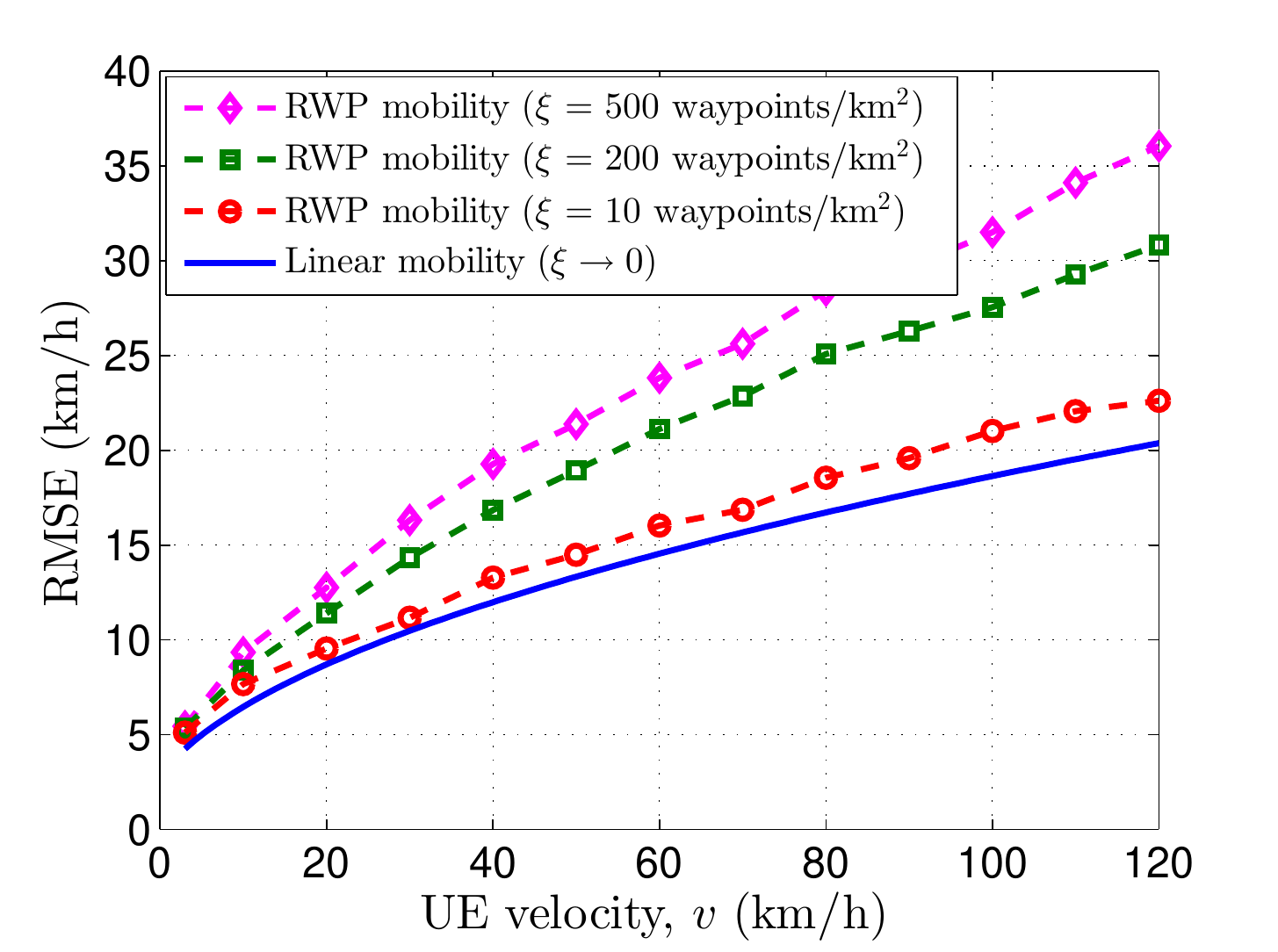}
\vspace{-7mm}\caption{}
\end{subfigure}
\begin{subfigure}[b]{0.49\textwidth}
\includegraphics[width=\textwidth]{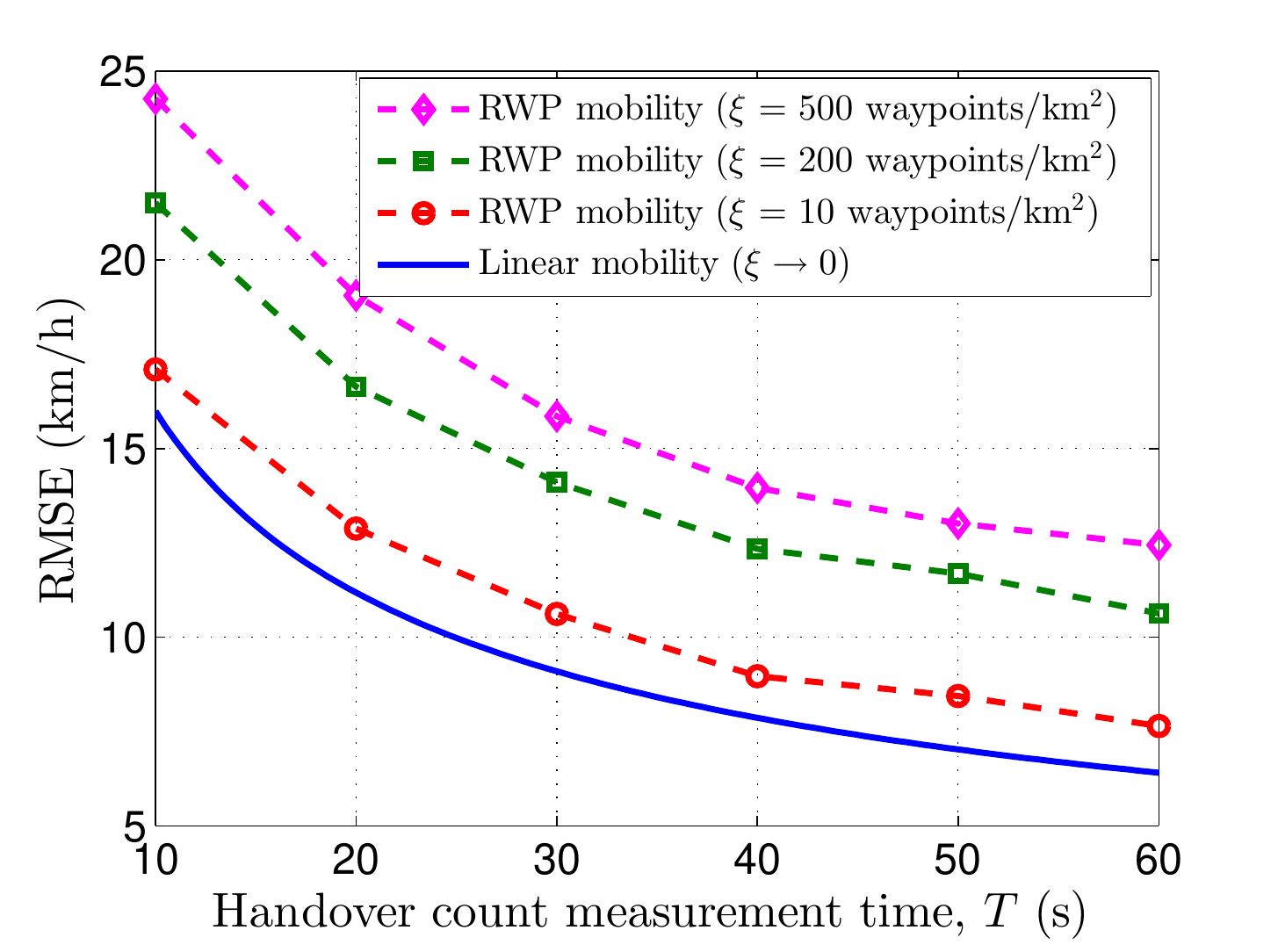}
\vspace{-7mm}\caption{}
\end{subfigure}
\vspace{-2mm}
\caption{RMSE performance of the velocity estimator with RWP mobility model; (a) RMSE versus UE velocity, with $T=12$~s and $\lambda=500$~SBSs/km$^2$; (b) RMSE versus handover count measurement time, with $v=60$~km/h and $\lambda=500$~SBSs/km$^2$.}
\label{fig:RMSE_RWPmobility}
\vspace{-5mm}
\end{figure}

\subsection{Mobility State Detection with Variable UE Velocity}
In this sub-section, we demonstrate the functionality of our MVU estimator with variable UE velocity, and the effect of handover-count measurement time on estimation accuracy. Consider an example in which a user is traveling in a train that is moving over a straight line trajectory in a downtown region. Assume that the density of small cells is $\lambda=1000$~SBSs/km$^2$, and handover counts are measured during regular intervals $T$. The actual velocity of UE in the train, estimated velocity and the detected mobility state are shown in Fig.~\ref{fig:VelocityTimeGraph} during a trip from one station to another station. It can be noticed that the accuracy of velocity estimation is better in Fig.~\ref{fig:VelocityTimeGraph}(b) with $T=30$~s, when compared to Fig.~\ref{fig:VelocityTimeGraph}(a) with $T=12$~s. Also, the probability of false alarm is smaller with $T=30$~s. However, longer measurement time intervals can lead to lower estimation accuracies if the velocity changes significantly within the measurement time interval.
\begin{figure}[htp]
\vspace{-3mm}
\centering
\begin{subfigure}[b]{0.49\textwidth}
\includegraphics[width=\textwidth]{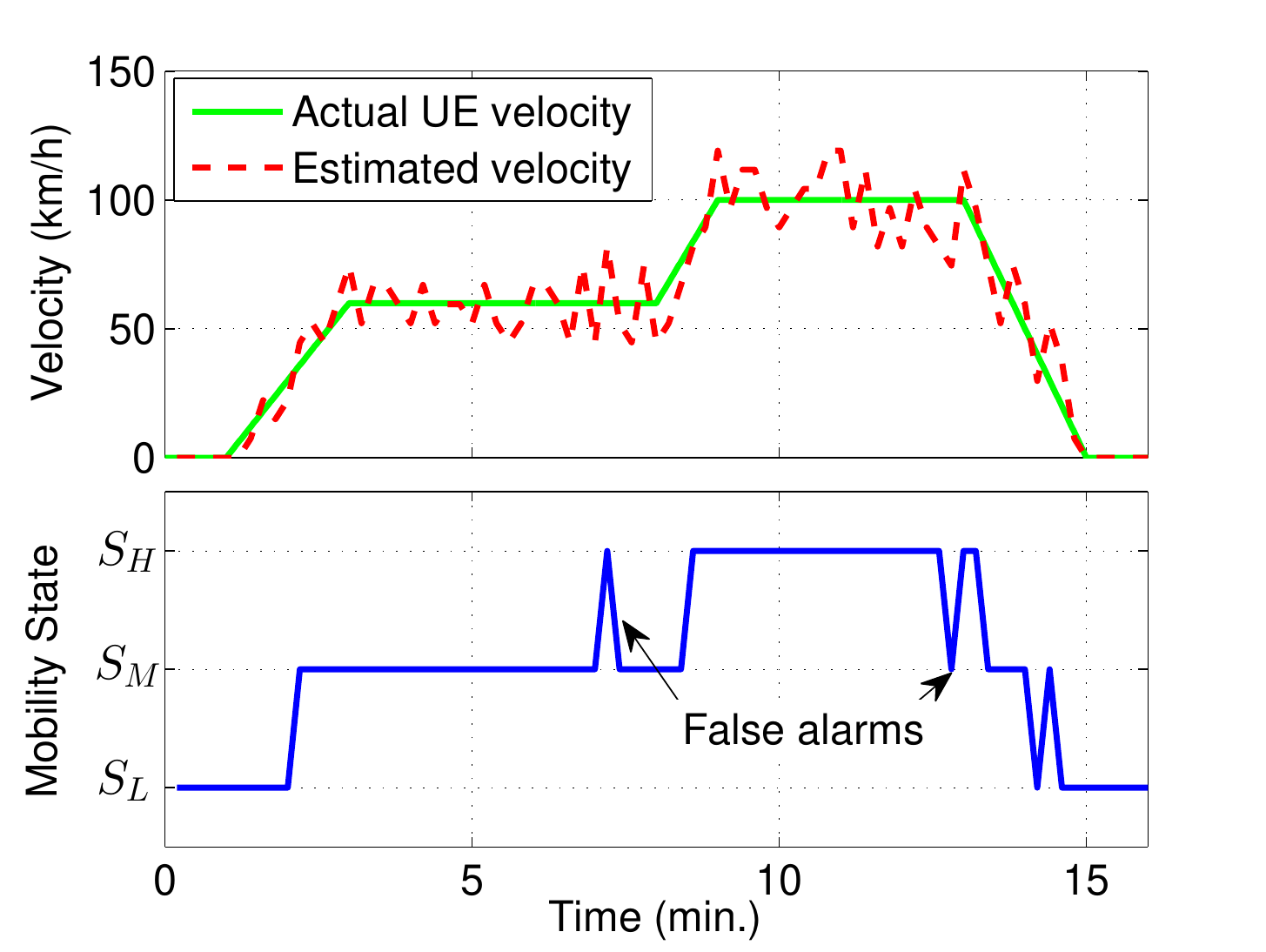}
\vspace{-7mm}\caption{}
\end{subfigure}
\begin{subfigure}[b]{0.49\textwidth}
\includegraphics[width=\textwidth]{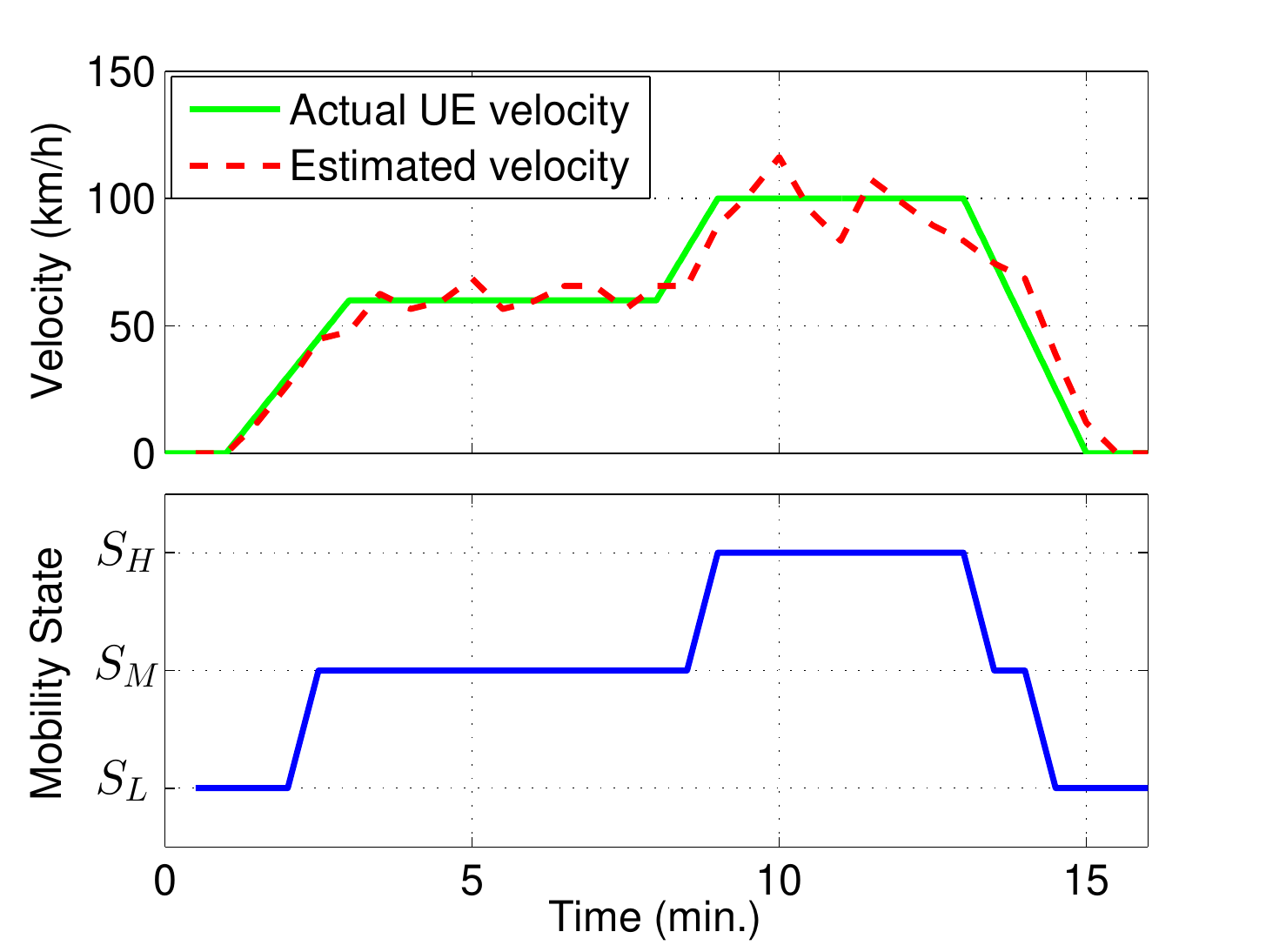}
\vspace{-7mm}\caption{}
\end{subfigure}
\vspace{-2mm}
\caption{Demonstration of velocity estimation and MSD with variable UE velocity; (a) $T=12$~s; (b) $T=30$~s.}
\label{fig:VelocityTimeGraph}
\vspace{-5mm}
\end{figure}

\subsection{Velocity Estimation in Clustered SBS Deployment}
In practical scenarios, SBSs are deployed more in places where more traffic demand is expected, which results into clustered deployment. In this subsection, we model the clustered deployment of SBSs using Matern-cluster process that can be realized using \emph{cluster centers} and \emph{cluster points}. The cluster centers form a Poisson process with intensity $\lambda_0$, and the cluster points (SBS locations) are located around each cluster center within a disc of radius $R$. The cluster points in these discs are realized using Poisson processes with intensity $\lambda_1$.

\begin{figure}[htp]
\vspace{-3mm}
\centering
\begin{subfigure}[b]{0.49\textwidth}
\includegraphics[width=\textwidth]{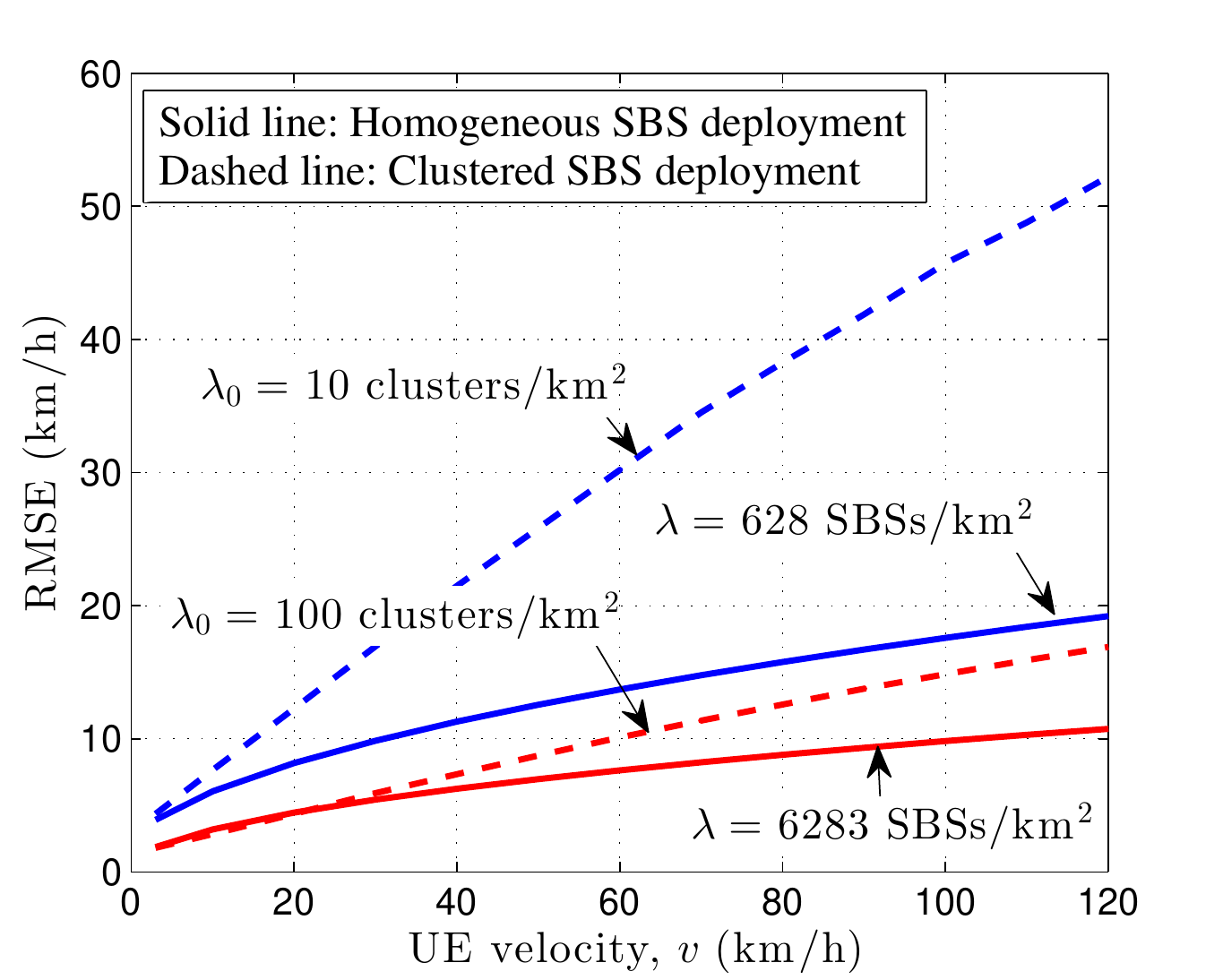}
\vspace{-7mm}\caption{}
\end{subfigure}
\begin{subfigure}[b]{0.49\textwidth}
\includegraphics[width=\textwidth]{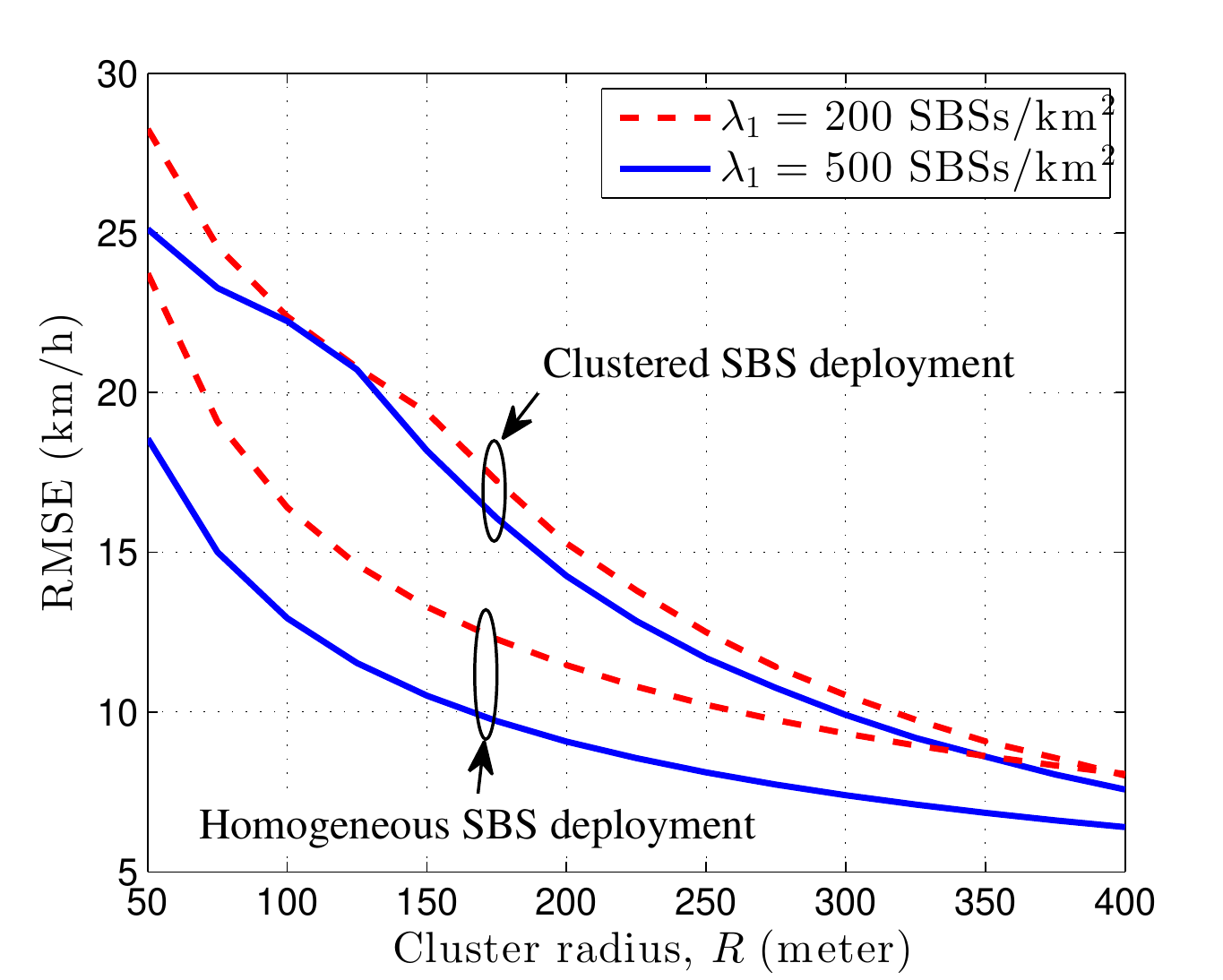}
\vspace{-7mm}\caption{}
\end{subfigure}
\vspace{-2mm}
\caption{RMSE of velocity estimation in homogeneous and clustered SBS deployment; (a) RMSE versus the velocity for different $\lambda_0$, with $\lambda_1=500$~SBSs/km$^2$, $R=200$~m, and $T=12$~s; (b) RMSE versus the cluster radius for different $\lambda_1$, with $\lambda_0=50$~clusters/km$^2$, $v=60$~km/h and $T=12$~s.}
\label{fig:RMSE_vs_Velocity_MaternCluster}
\vspace{-5mm}
\end{figure}
The RMSEs of velocity estimator in homogeneous and clustered deployments of SBSs are compared in Fig.~\ref{fig:RMSE_vs_Velocity_MaternCluster}. For fair comparison between the homogeneous and clustered deployments, the intensities of both Poisson and Matern-cluster processes were kept same. Given the parameters $\lambda_0, \lambda_1$ and $R$ of Matern-cluster process, the equivalent homogeneous intensity $\lambda$ for Poisson process can be found using $\lambda = \lambda_0 \lambda_1 \pi R^2$. In general, the RMSE is higher in clustered deployment of SBSs due to the heterogeneous nature of Matern-cluster process and the fact that a UE cannot obtain handover-count measurements outside the discs of clusters. It can be observed in Fig.~\ref{fig:RMSE_vs_Velocity_MaternCluster}(a) that the RMSE difference between homogeneous and clustered SBS deployments increases with increasing UE velocity $v$. The RMSE difference also increases significantly when $\lambda_0$ decreases, for example, from $100$ clusters/km$^2$ to $10$ clusters/km$^2$. In Fig.~\ref{fig:RMSE_vs_Velocity_MaternCluster}(b), it can be seen that the RMSE difference decreases with increasing $R$.

\subsection{Velocity Estimation with Matern Hardcore Process for SBS Locations}
In this sub-section, we model the SBS locations using Matern hardcore process (HCP), in which, the distance between any two points is greater than a hardcore distance $R_{\rm hc}$. The realization of Matern HCP involves generation of points using a homogeneous PPP, followed by a thinning procedure. For the Matern HCP of type~II, the thinning procedure involves associating a random mark to each point of the parent PPP, and a point is deleted if there exists another point within the hardcore distance $R_{\rm hc}$ with a smaller mark \cite{5934671}. To realize the HCP with intensity $\lambda_{\rm hc}$, the intensity of the parent PPP should be
\begin{align}
\lambda = \frac{-\ln(1-\lambda_{\rm hc}\pi R_{\rm hc}^2)}{\pi R_{\rm hc}^2},
\end{align}
provided that $R_{\rm hc}$ is smaller than $R_{\rm hc}^{\rm max}=1/\sqrt{\pi\lambda_{\rm hc}}$. We use $R_{\rm hc}=\rho/\sqrt{\pi\lambda_{\rm hc}}$, where $0\leq\rho<1$ is the randomness parameter. When $\rho=0$, $R_{\rm hc}$ is $0$ and the HCP is equivalent to PPP in which the points are random located as shown in Fig.~\ref{fig:MaternHCP_Realizations}(a). When, the $\rho$ increases to 0.99, the points are located more regularly as shown in Fig.~\ref{fig:MaternHCP_Realizations}(b).
\begin{figure}[htp]
\centering
\begin{subfigure}[b]{0.35\textwidth}
\includegraphics[width=\textwidth]{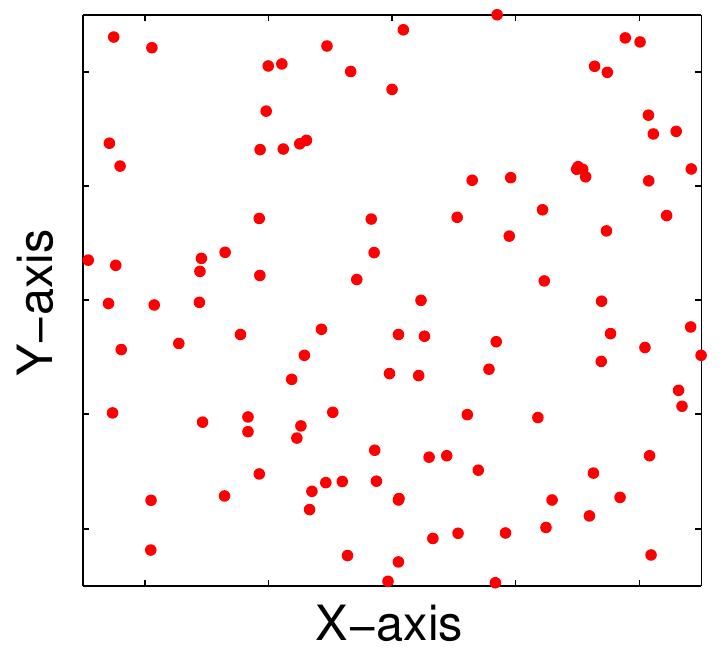}
\vspace{-7mm}\caption{}
\end{subfigure}
\begin{subfigure}[b]{0.35\textwidth}
\includegraphics[width=\textwidth]{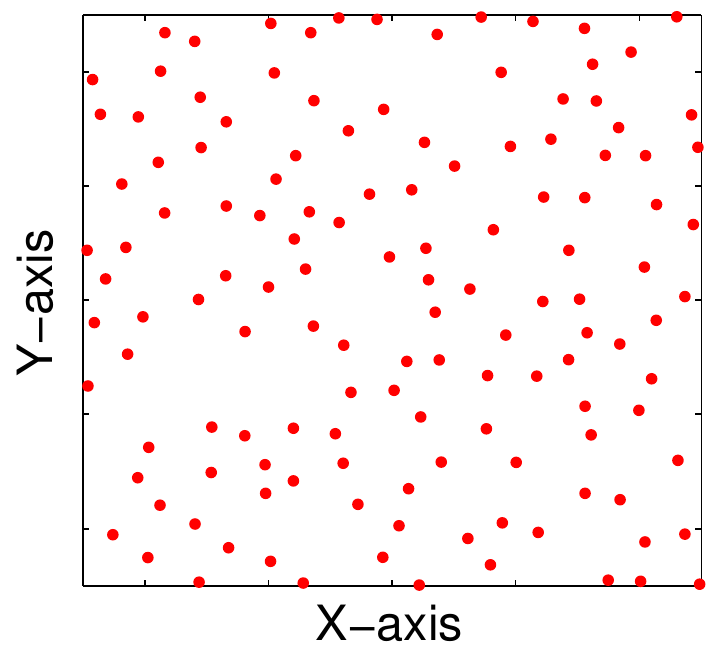}
\vspace{-7mm}\caption{}
\end{subfigure}
\vspace{-4mm}
\caption{Realizations of Matern HCP of type II with fixed $\lambda_{\rm hc}$ and different $\rho$ values; (a)~$\rho=0$; (b)~$\rho=0.99$.}
\label{fig:MaternHCP_Realizations}
\vspace{-3mm}
\end{figure}

The RMSE of the MVU velocity estimator against the variation in $\rho$ is shown in Fig.~\ref{fig:RMSE_vs_alpha_Matern_Type_II_HCP} for different velocities. The RMSE of the velocity estimator decreases as the $\rho$ increases, because the randomness in the SBS locations decreases with increasing $\rho$. Therefore, in scenarios where the SBSs are more uniformly placed, the velocity estimation will work more effectively.
\begin{figure}[t]
\vspace{-1mm}
\center
\includegraphics[width=0.49\textwidth]{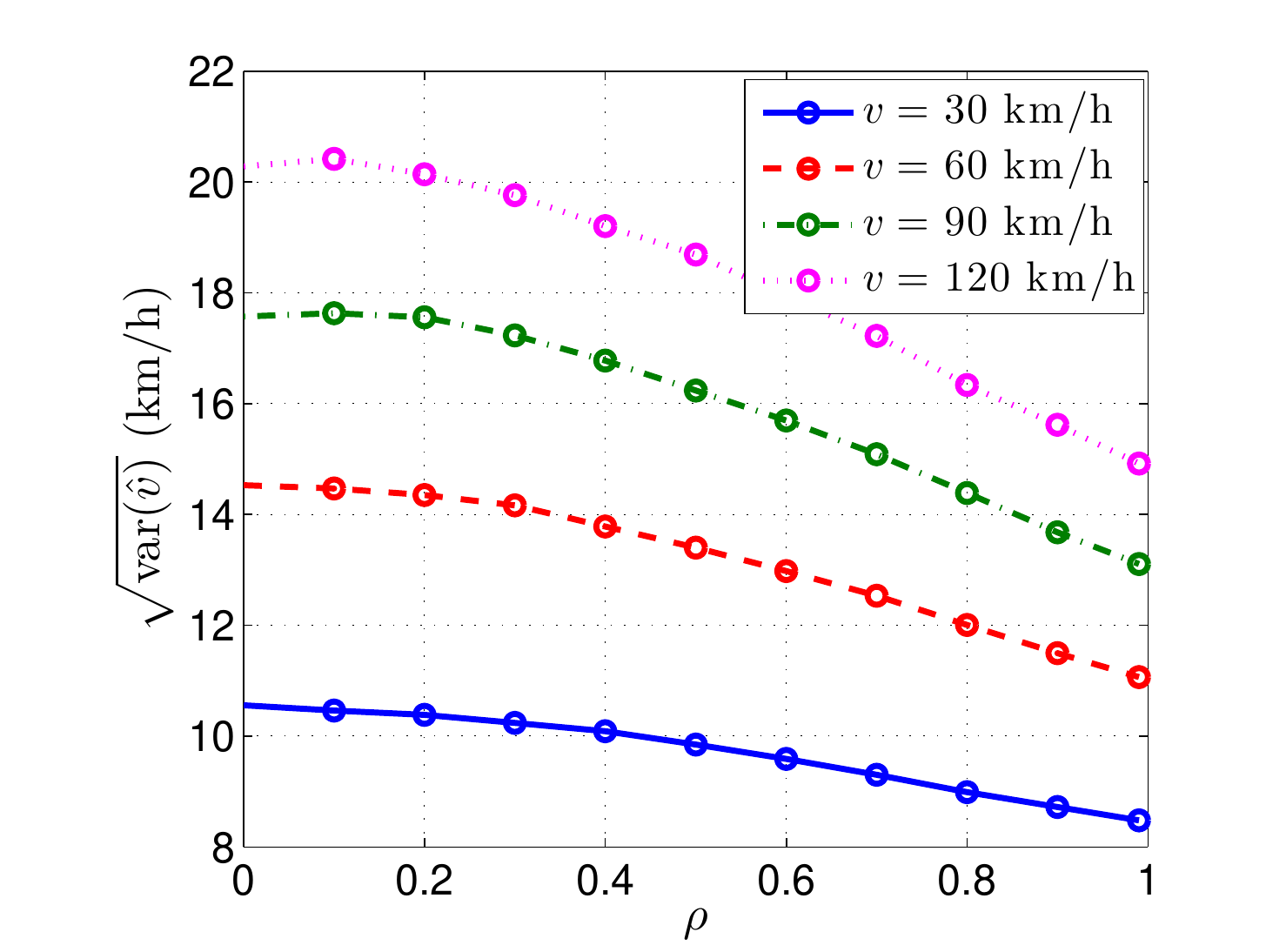}
\vspace{-2mm}
\caption{RMSE of velocity estimation with type II Matern HCP for the SBS locations.}
\label{fig:RMSE_vs_alpha_Matern_Type_II_HCP}
\end{figure}

\section{Conclusion}
\label{sec:Conclusion}
In this paper, we derived a handover-count based MVU estimator for the velocity of a UE, and derived optimum handover-count thresholds for MSD. We provided approximations to the PMF of handover-count using two types of standard distributions: gamma distribution and Gaussian distribution. We found that the PMF approximation using gamma distribution, even though complicated when compared to Gaussian distribution, provides a better fit to the PMF with a smaller MSE. On the other hand, the use of Gaussian distribution provides simpler and closed-form solutions to the PMF approximation. Subsequently, we derived the CRLB of velocity estimation and a MVU velocity estimator whose variance is approximately equal to the CRLB. The CRLB of velocity estimation decreases with the increasing time interval for counting the number of handovers, which shows the trade-off between the accuracy and the rapidness of velocity measurements.

The increasing density of small-cells in the future cellular networks is facilitates more accurate UE velocity estimation, since the CRLB decreases with the increasing SBS density. Moreover, the probability of MSD also increases with the increasing SBS density. The results with RWP mobility model showed that the accuracy of the velocity estimator decreases with the increasing randomness in UE's mobility. On the other hand, results with Matern-cluster process showed that the estimation accuracy is poorer with clustered SBS deployment when compared to the homogeneous SBS deployment. Future research directions may include: path prediction of UEs, and development of algorithms to dynamically adjust the measurement time interval depending on the past samples of estimated velocity.

\appendices
\section{Approximating the $\alpha$ and $\beta$ Parameters of Gamma Distribution}
\label{app:AlphaBetaStudy}
In this appendix, we derive the expressions for $\alpha$ and $\beta$ parameters through a heuristic approach to minimize the MSE between the approximation $f^{\rm g}_{H}(h)$ and the PMF $f_H(h)$ of the handover count. First, we briefly describe a property of PPP with respect to the scaling of simulation space which will be used for the derivation of $\alpha$ and $\beta$ parameters.
\begin{figure}[t]
\centering
\begin{subfigure}[b]{0.49\textwidth}
\includegraphics[width=\textwidth]{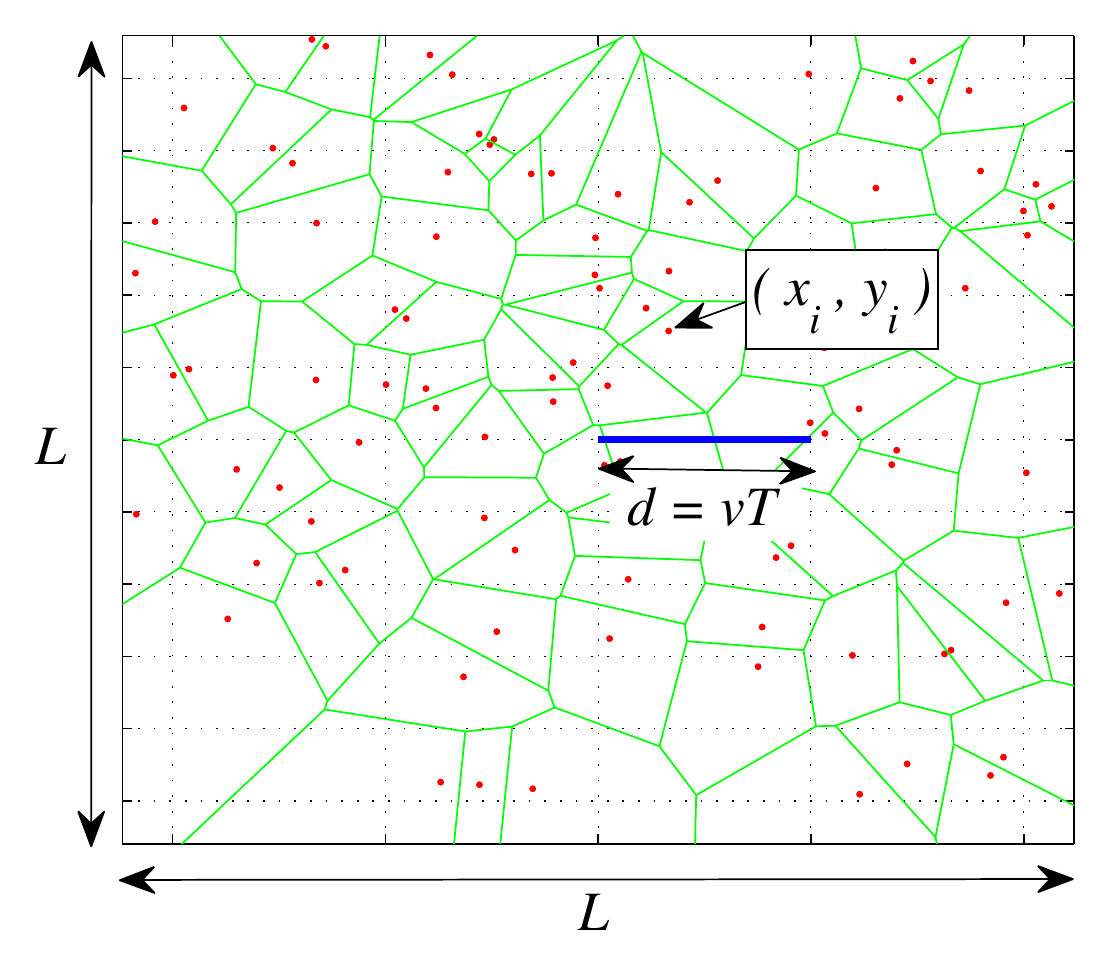}
\vspace{-7mm}\caption{}
\end{subfigure}
\begin{subfigure}[b]{0.49\textwidth}
\includegraphics[width=\textwidth]{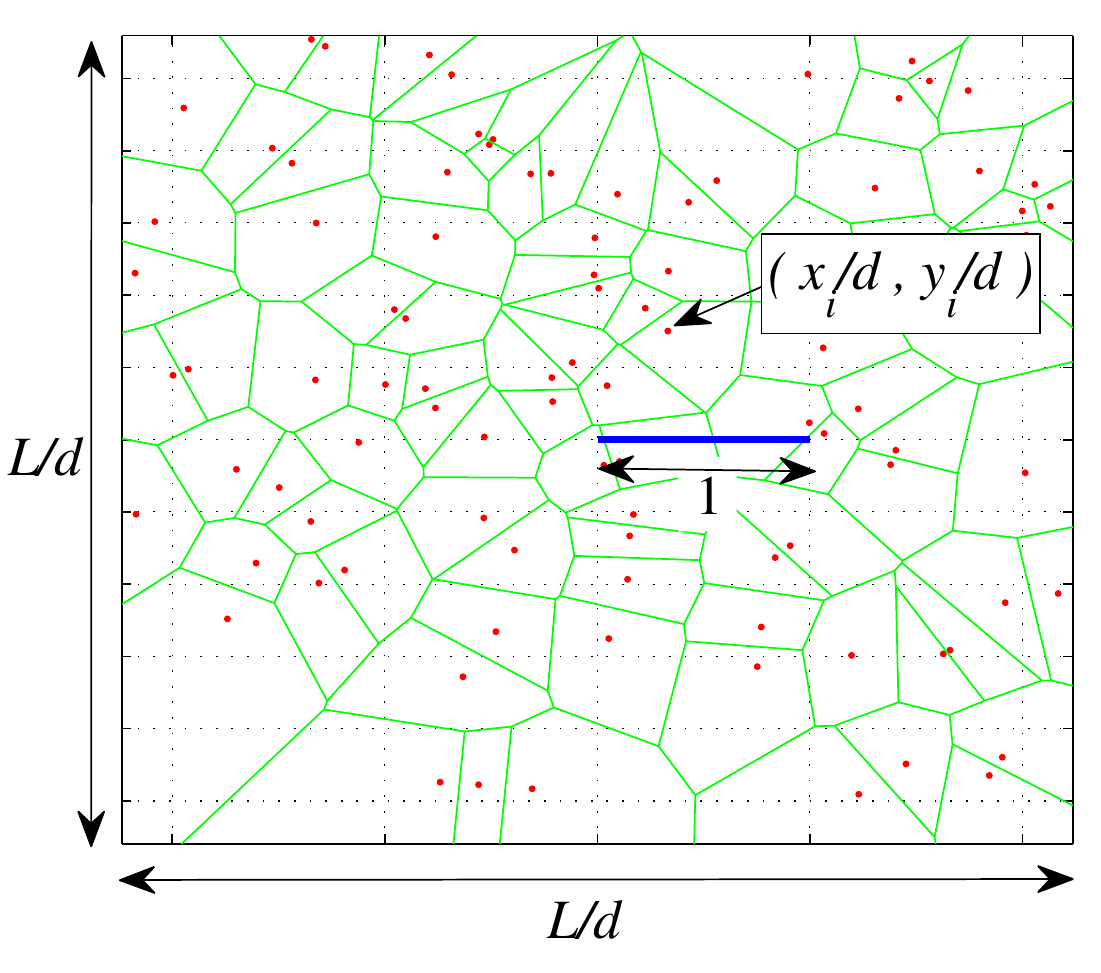}
\vspace{-7mm}\caption{}
\end{subfigure}
\vspace{-2mm}
\caption{Illustration for scaling down the simulation space.}
\label{fig:NormalizeSpaceIllustration}
\vspace{-5mm}
\end{figure}

Consider the square shaped simulation space shown in Fig.~\ref{fig:NormalizeSpaceIllustration}(a) with an area $L^2$. The SBS locations are shown in red dots, and the UE trajectory is shown as blue line of length $d=vT$. The intensity of SBSs can be expressed as $\lambda = N_{\rm avg}/L^2$, where $N_{\rm avg}$ is the mean number of SBSs in the simulation area. Assume that the simulation space is scaled down by the UE travel distance $d$, i.e., all the coordinates are scaled down by $d$, for example, coordinates of a SBS location $(x_i, y_i)$ will be scaled down to $(\frac{x_i}{d}, \frac{y_i}{d})$ as shown in Fig.~\ref{fig:NormalizeSpaceIllustration}(b). Therefore, the simulation area will also be scaled down to $\left(L/d\right)^2$. However, the mean number of SBSs $N_{\rm avg}$ remains the same, and hence the intensity of SBSs can be expressed as $\lambda' = N_{\rm avg} d^2/L^2 = \lambda d^2$. Here, the term $\lambda'$ is a function of $\lambda$ and $d$, therefore, the statistics of the handover count $H$ can be expressed in terms of $\lambda'$ alone. For example, the mean number of handovers in \eqref{eq:MeanNo_HOs} can be expressed as $E[H] = 4\sqrt{\lambda'}/\pi$. Similarly, the parameters $\alpha$ and $\beta$ can also be approximated in terms of $\lambda'$.

We obtained histogram of the handover count $H$ through extensive simulations in Matlab with a range of values for the SBS density $\lambda=\{100, 200, 300,...,10000\}$ SBSs/km$^2$ and the distance $d=\{0.033, 0.1, 0.2, 0.4\}$ km. For each combination of $\lambda$ and $d$, we obtained one million samples of $H$ for constructing the PMF $f_H(h)$. Then, we used the curve fitting tool in Matlab to obtain the $\alpha$ and $\beta$ parameters that provides a best fit of $f^{\rm g}_H(h)$ in \eqref{eq:fg_hat_final} to the PMF $f_H(h)$. By studying the characteristics of $\alpha$ and $\beta$ parameters with respect to the variations in $\lambda$ and $d$, we formulated
\begin{align}
\alpha &= 2.7 + 4 \sqrt{\lambda'} = 2.7 + 4 d\sqrt{\lambda},\label{eq:AlphaApprox_2}\\
\beta &= \pi + \frac{0.8}{0.38 + \sqrt{\lambda'}} = \pi + \frac{0.8}{0.38 + d\sqrt{\lambda}}.\label{eq:BetaApprox_2}
\end{align}
The accuracy of these approximations are justified in Fig.~\ref{fig:AlphaBetaApprox} where the theoretical expressions \eqref{eq:AlphaApprox_2} and \eqref{eq:BetaApprox_2} are shown to tightly overlap with the plots obtained through simulations.
\begin{figure}[htp]
\vspace{-3mm}
\center
\includegraphics[width=3in]{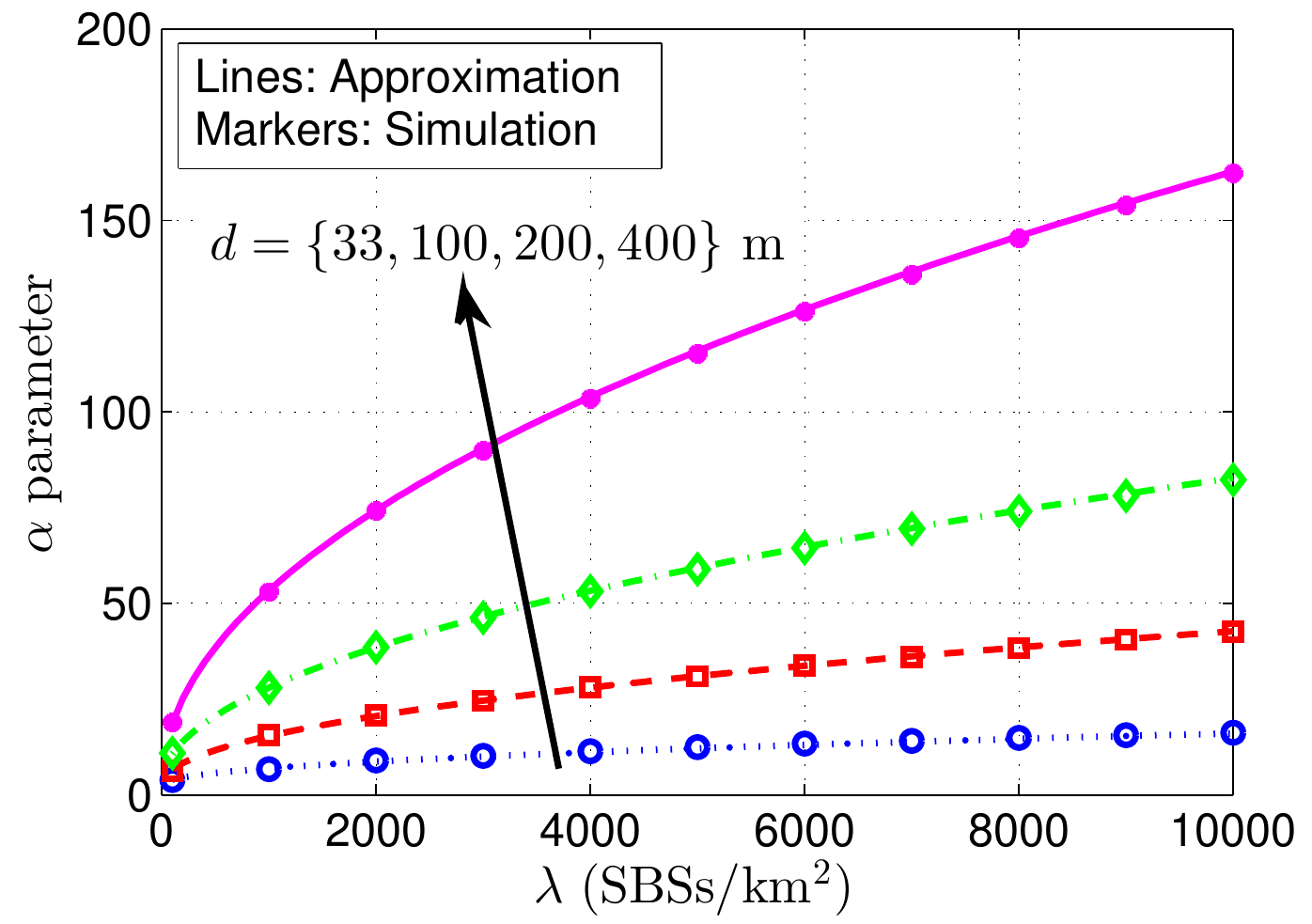}
\includegraphics[width=3in]{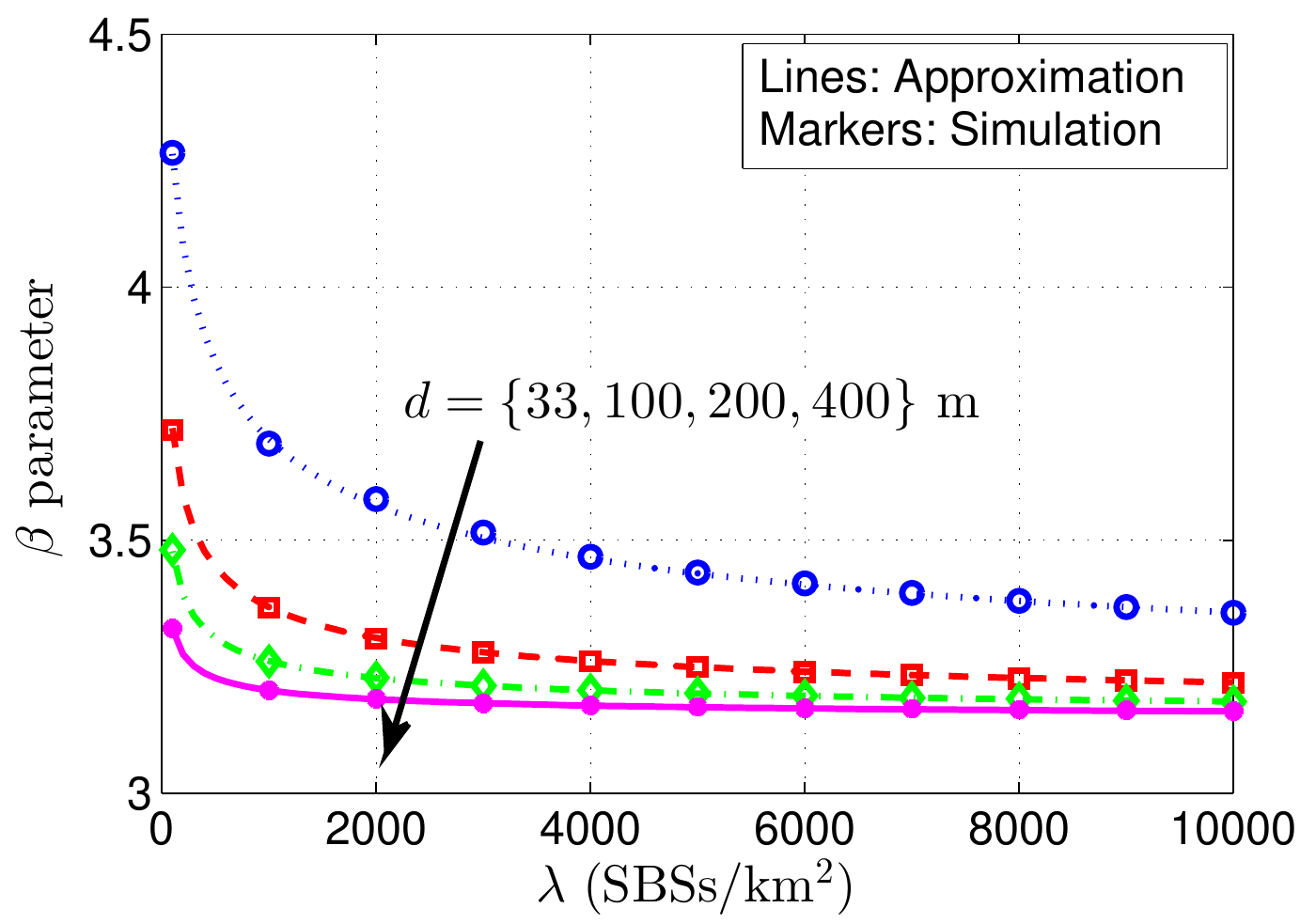}
\vspace{-3mm}
\caption{Approximations of the $\alpha$ and $\beta$ parameters.}
\label{fig:AlphaBetaApprox}
\end{figure}

\section{Approximating the $\mu$ and $\sigma^2$ Parameters of Gaussian Distribution}
\label{app:MuSigma2Study}
In this appendix, we derive the expressions for $\mu$ and $\sigma^2$ in \eqref{eq:fn_hat_final} to minimize the MSE between the approximation $f^{\rm n}_{H}(h)$ and the PMF $f_H(h)$ of the handover count. We equate $\mu$ to the mean of handover count expressed in \eqref{eq:MeanNo_HOs} as
\begin{align}
\mu = \frac{4vT\sqrt{\lambda}}{\pi} = \frac{4d\sqrt{\lambda}}{\pi}.
\end{align}
Next, we obtain an expression for $\sigma^2$ using a method similar to the one explained in Appendix~\ref{app:AlphaBetaStudy}. We constructed the PMF $f_H(h)$ through simulations, and then used curve fitting tool in Matlab to obtain the $\sigma^2$ parameter that provides a best fitting of $f^{\rm n}_H(h)$ in \eqref{eq:fn_hat_final} to the PMF $f_H(h)$. By studying the characteristics of $\sigma^2$ parameter with respect to the variations in $\lambda$ and $d$, we formulated
\begin{align}
\sigma^2 &= 0.07 + 0.41 \sqrt{\lambda'} = 0.07 + 0.41 d \sqrt{\lambda}. \label{eq:Sigma2Approx_1}
\end{align}
The accuracy of this approximation is justified in Fig.~\ref{fig:Sigma2Approx} where the expression \eqref{eq:Sigma2Approx_1} is plotted in comparison with the plots obtained through simulations.
\begin{figure}[htp]
\center
\includegraphics[width=3in]{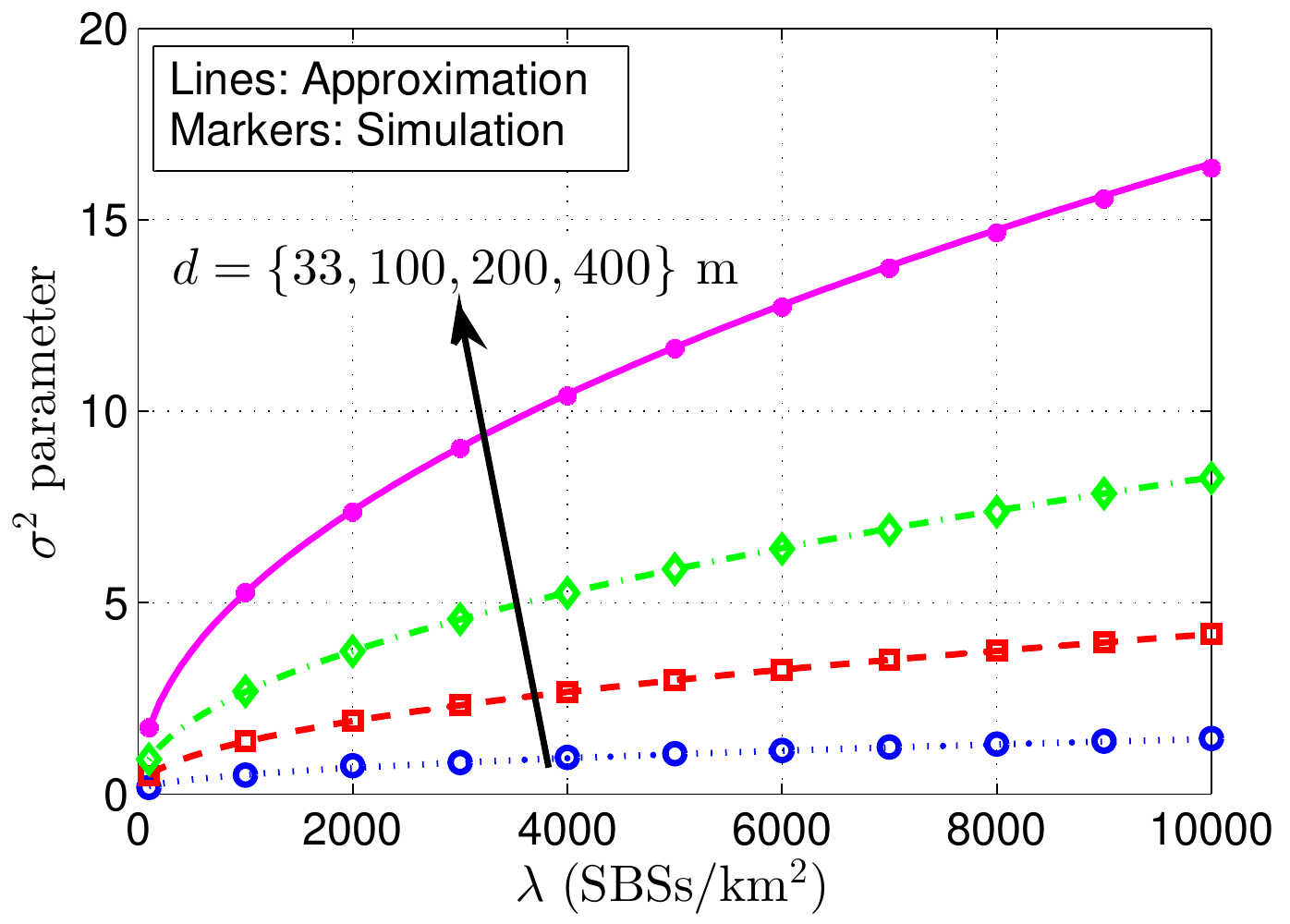}
\vspace{-2mm}
\caption{Approximation of the $\sigma^2$ parameter.}
\label{fig:Sigma2Approx}
\end{figure}

\section{Derivation of CRLB with Gamma Approximation for Handover Count PMF}
\label{app:CRLB_Gamma}
In this appendix, CRLB for velocity estimation will be derived by considering the handover count PMF in \eqref{eq:fg_hat_final} that was derived using the gamma distribution. Taking the logarithm of PMF will provide the log-likelihood function, which will then be differentiated with respect to $v$. Finally, the derivative of log-likelihood function will be used to obtain the CRLB expression. Consider the handover count PMF,
\begin{align}
f^{\rm g}_H(h; v) = \frac{\Gamma\big(\alpha,\beta h, \beta(h+1)\big)}{\Gamma(\alpha)}, \label{eq:fg_hat_final_1}
\end{align}
where, the $\alpha$ and $\beta$ parameters are given by \eqref{eq:AlphaApprox} and \eqref{eq:BetaApprox}, respectively. By taking the logarithm of \eqref{eq:fg_hat_final_1} and differentiating with respect to $v$, we get
\begin{align}
\frac{\partial}{\partial v} \log f^{\rm g}_H(h; v) =\frac{\partial}{\partial v} \log \Gamma\big(\alpha,\beta h, \beta(h+1)\big) - \frac{\partial}{\partial v} \log \Gamma(\alpha). \label{eq:DiffLogPMF}
\end{align}

Consider the first term in the RHS of \eqref{eq:DiffLogPMF}. Let, $z_1 = \beta h$ and $z_2 = \beta (h+1)$. Then, we have
\begin{align}
&\frac{\partial}{\partial v} \log \Gamma\big(\alpha,\beta h, \beta(h+1)\big) = \frac{1}{\Gamma(\alpha, z_1, z_2)} \frac{\partial}{\partial v}\Gamma(\alpha, z_1, z_2),\nonumber\\
&\hspace{1cm}= \frac{1}{\Gamma(\alpha, z_1, z_2)} \left[\frac{\partial}{\partial \alpha}\Gamma(\alpha, z_1, z_2) \frac{{\rm d}\alpha}{{\rm d}v} + \frac{\partial}{\partial z_1}\Gamma(\alpha, z_1, z_2)\frac{{\rm d}z_1}{{\rm d}v} + \frac{\partial}{\partial z_2}\Gamma(\alpha, z_1, z_2)\frac{{\rm d} z_2}{{\rm d}v}\right]. \label{eq:DiffLogPMF4}
\end{align}
Each of the differentials in \eqref{eq:DiffLogPMF4} can be derived to be
\begin{align*}
&\frac{\partial}{\partial \alpha}\Gamma(\alpha, z_1, z_2) = \frac{{_2F_2}(\alpha,\alpha;\alpha+1,\alpha+1;-z_1) z_1^\alpha}{\alpha^2} - \frac{{_2F_2}(\alpha,\alpha;\alpha+1,\alpha+1;-z_2) z_2^\alpha}{\alpha^2}\nonumber\\
&\hspace{8.4cm} -\gamma(\alpha,z_1)\log(z_1) +\gamma(\alpha,z_2)\log(z_2), \\
&\hspace{1cm}\frac{\partial}{\partial z_1}\Gamma(\alpha, z_1, z_2) = -e^{-z_1} z_1^{\alpha-1}, \hspace{1cm}\frac{\partial}{\partial z_2}\Gamma(\alpha, z_1, z_2) = e^{-z_2} z_2^{\alpha-1},\\
&\hspace{1cm}\frac{{\rm d}\alpha}{{\rm d}v} = 4 T \sqrt{\lambda}, \hspace{1cm}\frac{{\rm d}z_1}{{\rm d}v} = \frac{-0.8 h T \sqrt{\lambda}}{(0.38+v T \sqrt{\lambda})^2}, \hspace{1cm} \frac{{\rm d} z_2}{{\rm d}v} = \frac{-0.8 (h+1)T \sqrt{\lambda}}{(0.38+v T \sqrt{\lambda})^2},
\end{align*}
where, $\gamma(\alpha,x) = \int_0^x t^{\alpha-1}e^{-t} {\rm d}t$ is the lower incomplete gamma function, and ${_2F_2}(a_1,a_2;b_1,b_2;z)$ is generalized hypergeometric function which can be expressed as
\begin{align}
{_2F_2}(a_1,a_2;b_1,b_2;z) = \sum_{k=0}^\infty\frac{(a_1)_k (a_2)_k}{(b_1)_k (b_2)_k} \frac{z^k}{k!},
\end{align}
where, $(a)_0 = 1$ and $(a)_k = a(a+1)(a+2)...(a+k-1)$, for $k \geq 1$.

Now, consider the second term in RHS of \eqref{eq:DiffLogPMF}, $\frac{\partial}{\partial v} \log \Gamma(\alpha) = \frac{\partial \log \Gamma(\alpha)}{\partial \alpha} \frac{\partial \alpha}{\partial v}$. It is well known that the logarithmic derivative of the gamma function is the digamma function $\psi(\cdot)$. Therefore,
\begin{align}
\frac{\partial}{\partial v} \log \Gamma(\alpha) &= \psi(\alpha) \frac{\partial \alpha}{\partial v} = \psi(\alpha) \frac{\partial}{\partial v} (2.7 + 4vT\sqrt{\lambda}) = 4T\sqrt{\lambda}\ \psi(\alpha). \label{eq:SecondTerm}
\end{align}
Using the equations \eqref{eq:DiffLogPMF4}-\eqref{eq:SecondTerm}, we can express \eqref{eq:DiffLogPMF} as:
\begin{align}
\frac{\partial}{\partial v} \log f_H^{\rm g}(h; v) =& \frac{4 T \sqrt{\lambda} \beta^\alpha}{\alpha^2 \Gamma\big(\alpha, \beta h, \beta(h+1)\big)} \Big[h^\alpha {_2F_2}(\alpha,\alpha;\alpha+1,\alpha+1;-\beta h) \nonumber\\
&\hspace{4.5cm}- (h+1)^\alpha {_2F_2}\big(\alpha,\alpha;\alpha+1,\alpha+1;-\beta(h+1)\big) \Big]\nonumber\\
&- \frac{4 T \sqrt{\lambda}}{\Gamma\big(\alpha, \beta h, \beta(h+1)\big)}\Big[\gamma(\alpha,\beta h)\log(\beta h) -\gamma\big(\alpha,\beta(h+1)\big)\log\big(\beta(h+1)\big)\Big] \nonumber\\
&+ \frac{0.8T \sqrt{\lambda}\beta^{\alpha-1}e^{-\beta h}\left[h^\alpha - e^{-\beta} (h+1)^\alpha \right]}{\Gamma\big(\alpha, \beta h, \beta(h+1)\big)(0.38+v T \sqrt{\lambda})^2} - 4T\sqrt{\lambda}\ \psi(\alpha). \label{eq:DiffLogPMF_2}
\end{align}
By squaring \eqref{eq:DiffLogPMF_2} and evaluating the expectation over $H$, we obtain the \emph{Fisher information} as
\begin{align}
I(v) = \mathbb{E}\left[\left(\frac{\partial \log f_H^{\rm g}(H; v)}{\partial v}\right)^2\right].
\end{align}
By taking the reciprocal of Fisher information, we can express the CRLB for $\hat{v}$ as in \eqref{eq:CRLB}.

\section*{Acknowledgment}
This research was supported in part by the U.S. National Science Foundation under the Grant CNS-1406968.

\singlespacing
\bibliographystyle{IEEEtran}
\bibliography{Bibliography}

\end{document}